\documentclass[a4paper, 12pt]{article}
\usepackage{amsfonts,amsmath,amsthm}
\usepackage{graphicx}
\usepackage{authblk}
\usepackage{hyperref}
\usepackage{microtype}
\usepackage{lipsum}
\usepackage[section]{placeins}
\usepackage{afterpage}
\usepackage{capt-of}% or use the larger `caption` package
\usepackage{pifont}% http://ctan.org/pkg/pifont
\usepackage{bm}
\usepackage{pdfpages}
\usepackage{mdframed}
\usepackage{multirow}
\usepackage[ruled]{algorithm2e} % For algorithms

\usepackage{algorithmic}
\usepackage{multicol}
\usepackage{multirow}
\usepackage{makecell}
\usepackage{rotating}
\usepackage{arydshln}

\usepackage{soul}
\usepackage{xspace}

\DeclareMathOperator*{\argmin}{arg\,min}

\newtheorem{thm}{Theorem}
\newtheorem{lem}[thm]{Lemma}

\newtheorem{rem}[thm]{Remarks}

\def\vol{\mbox{vol}}

\def\RR{{\mathbb R}}
\def\R{{\mathbb R}}

\def\Sd{{\mathcal S_{d-1}}}

\def\cK{\mathcal{K}}

\usepackage{amsmath}
\usepackage[utf8]{inputenc}
\usepackage{soul}

\newcommand{\volesti}{\href{https://github.com/GeomScale/volume_approximation}{\textcolor{blue}{\texttt{volesti}}}\xspace}
\newcommand{\eigen}{\href{http://eigen.tuxfamily.org}{\textcolor{blue}{\texttt{eigen}}}\xspace}
\newcommand{\boost}{\href{http://boost.com}{\textcolor{blue}{\texttt{boost}}}\xspace}

\newcommand{\logconcdead}{\href{https://CRAN.R-project.org/package=LogConcDEAD}{\textcolor{blue}{\texttt{LogConcDEAD}}}\xspace}

\usepackage{geometry}
\begin{document}
\date{\quad}
\title{Randomized geometric tools for anomaly detection in stock markets}
\author[1]{Cyril Bachelard}
\author[2]{Apostolos Chalkis}
\author[3,2]{Vissarion Fisikopoulos}
\author[4,2]{\\Elias Tsigaridas}

\affil[1]{Faculty of Business and Economics (HEC)\\ Department of Operations\\ University of Lausanne, Switzerland}
\affil[2]{GeomScale org}
\affil[3]{National \& Kapodistrian University of Athens, Greece}
\affil[4]{Inria Paris and  IMJ-PRG\\ Sorbonne Universit\'e and Paris Universit\'e}

\maketitle

\begin{abstract}
We propose novel randomized geometric tools to detect 
low-volatility anomalies in stock markets; a principal problem in financial economics.
Our modeling of the (detection) problem results in sampling 
and estimating the (relative) volume of geodesically non-convex and non-connected spherical patches
that arise by intersecting a non-standard simplex with a sphere.
To sample, we introduce two novel Markov Chain Monte Carlo (MCMC) algorithms that exploit the geometry of the problem and employ state-of-the-art continuous geometric random walks (such as Billiard walk and Hit-and-Run) adapted on spherical patches.
To our knowledge, this is the first geometric formulation and MCMC-based analysis of the volatility puzzle in stock markets.
We have implemented our algorithms in C++ (along with an R interface)
and we illustrate the power of our approach by performing extensive experiments on real data. 
Our analyses provide accurate detection and new insights into the distribution of portfolios’ performance characteristics. Moreover, we use our tools to show that classical methods for low-volatility anomaly detection in finance form bad proxies that could lead to misleading or inaccurate results. 
\end{abstract}

\newpage
% Paper body
\section{Introduction}
\label{sec:intro} 

We consider two fundamental problems from two arguably distant disciplines: Computational Geometry and Financial Economics. The geometrical problem involves the computation of volume of and sampling from non-convex and possibly disconnected spherical patches arising by the intersection of a non-standard simplex with a sphere. The absence of convexity and the presence of many connected components make the problem very challenging from an algorithmic and implementation point of view. The geometric representation has a financial interpretation, which is the set of portfolios, i.e., investments in a collection of stocks, having a certain risk level. Thus, our motivation to solve this geometrically hard problem, aside from having an interest in its own right, stems from a concrete financial question about the feasible space of investable portfolios obeying certain regulatory and risk related constraints. The fundamental Economic question that we eventually address is that of the relation between risk and return: do assets with higher risk levels provide higher returns, as suggested by economic theory, or is there an empirically observable anomaly? The question is a long-standing one and we give a glimpse into its history in the next section. 

Our analysis and contribution to the topic consist of a differentiated perspective from the existing literature. The usual procedure is to form a small number of portfolios by first sorting stocks according to their historical volatility and grouping them into portfolios (typically equally weighted) with increasing riskiness in order to then assess their relative empirical performance. Instead of the simple heuristic approach to group stocks into risk buckets, we consider all possible combinations of stocks from a given investable universe, i.e., the entire space of feasible portfolios having a certain risk level (defined in terms of volatility). 
To achieve this goal we use sampling based on geometric random walks, i.e., Markov Chain Monte Carlo (MCMC) procedures, and compare their resulting empirical risk-return characteristic to those of portfolios sampled from other risk level sets. 
Hence, instead of comparing the descriptive performance statistics of a single portfolio to represent a risk level, we investigate the joint distribution of risk and return of the parameters of the average portfolio having a certain risk level. 
Analyzing two different stock markets (the USA and Europe) and subsets of the two markets distinguishing between company size (small versus large) and sector belongingness (cyclical versus defensive) we find that the risk-return profiles of commonly employed equally weighted quintile portfolios formed by sorting stocks according to their historical volatility range anywhere in the bivariate distribution of realized risk and return. 
As such, we argue that one should be careful to draw conclusions drawn from the sorting-based point estimates. 
However, the analysis of sets of portfolios with given risk level, thus abstracting from the problem of a specific weighting scheme and the missing consideration of correlation structures among stocks, shows that the resulting risk-return cluster does indeed support the hypothesis of an anomaly, albeit less pronounced than when using the quintile approach. The sampling-based approach allows us to visualize how the distribution of risk and return widens with increasing ex-ante volatility. Having knowledge about the distribution of performance statistics is then valuable for statistical inference and significance testing, particularly in the context of financial data which display time-series structures and are non-Gaussian.

%MCMC sampling of the asset weights spaces is difficult since the geometric bodies in question are non-convex and sometimes disconnected, making it necessary to compute the relative volume of the components in order to obtain samples from the targeted stationary distribution, which in our case, is the uniform distribution. We introduce two geometric random walks, namely the Great Cycle Walk (GCW) and the Reflective Great Cycle Walk (ReGCW) as generalizations of Hit-and-Run~\cite{Smith93} and Billiard Walk~\cite{Gryazina14} on a surface patch. \footnote{vis:should this go to 1.2?}

\subsection{Financial background}

It has long been recognized that the capital asset pricing model (CAPM), a cornerstone of financial economic theory and the workhorse model of classical capital market theory, independently developed by \cite{bib:Sharpe1964} and \cite{bib:Lintner1965} and \cite{bib:Mossin1966}, does not do justice in explaining the complexity of real world market dynamics. Contrary to the equilibrium model’s predicted simple positive linear relation between risk and expected return, higher risk is not generally rewarded with higher return in global stock markets. According to the CAPM, the return one should expect from an investment depends solely on the riskiness of the investment relative to a single factor which is the overall market. Investments which bear higher risk than the market portfolio should pay out a higher return in expectation, i.e., a risk premia.  However, Haugen and Heins \cite{bib:HaugenHeins1975} were the first to recognize that risk does not generate a special reward following the early warning signs from \cite{bib:BlackJensenScholes1972}, \cite{bib:MillerScholes1972} and \cite{bib:FamaMacBeth1973}. Their finding has subsequently been confirmed by \cite{bib:FamaFrench1992} and \cite{bib:Black1993}.

Further studies have found a wealth of anomalies, i.e., systematic and persistent deviations of empirical observations from model prediction. Prominent examples include firm size \cite{bib:Banz1981} (stocks with lower market capitalization tend to outperform stocks with a higher market capitalization in the future), value \cite{bib:RosenbergReidLanstein1985} (stocks that have a low price relative to their fundamental value, commonly tracked via accounting ratios like price to book or price to earnings outperform high-value stocks), momentum \cite{bib:JegadeeshTitman1993} (stocks that have outperformed in the past tend to exhibit strong returns going forward) or quality \cite{bib:AsnessFrazziniPedersen2019} (stocks which have low debt, stable earnings, consistent asset growth, and strong corporate governance, commonly identified using metrics like return to equity, debt to equity, and earnings variability).

Building upon the anomalous findings, the original single-factor CAPM has then been augmented with other factors besides the market, namely size and value \cite{bib:FamaFrench1992}, size value and momentum \cite{bib:Carhart1997} and size, value, and two quality factors \cite{bib:FamaFrench2015}. Nevertheless, despite the wealth of documented anomalies and cited extensions, the CAPM has shown great resilience vis-à-vis a transition to an alternative paradigm. This is particularly surprising in the light of a large body of literature subsumed under the term low-risk anomaly or low-volatility anomaly which directly attacks the very core of the CAPM by showing that, even after controlling for other factors as done by the latest CAPM extensions, (i) low-risk companies outperform high-risk companies and that (ii) low-risk portfolios produce higher risk-adjusted returns than capitalization weighted benchmarks.

In particular, \cite{bib:Falkenstein1994} found that, when controlling for size, the relation between risk and return gets reversed. Further, the outperformance of low-volatility stocks compared to high-volatility stocks has been shown to be robust among different markets, industries and sub-periods \cite{bib:BlitzvanVliet2007}, \cite{bib:BlitzPangvanVliet2013}, \cite{bib:BakerHaugen2012}, \cite{bib:vanVlietdeKoning2017}, \cite{bib:BlitzvanVlietBaltussen2019}). \cite{bib:Walkshausl2014}. 

Note that in the CAPM, the risk is defined in terms of market beta, i.e., the ratio between an asset’s volatility (standard deviation of returns) and the volatility of a market portfolio times their correlation (or the ratio of the asset’s co-variation with the market portfolio over the variance of the market portfolio). The low-risk anomaly was found to be present irrespective of the employed measure of risk, be it beta or volatility (\cite{bib:BlitzvanVliet2007}, \cite{bib:BakerBradleyWurgler2011}, \cite{bib:FrazziniPedersen2014}, \cite{bib:LiuStambaughYuan2018}, \cite{bib:AsnessEtAl2020}). 

The typical approach pursued by the studies on the low-volatility anomaly is to sort stocks according to their historical volatility and to form portfolios, either weighted equally or proportional to market capitalization, within quantiles of volatility levels. The process is then repeated on a monthly or quarterly basis, thus giving rise to, say, five (quintile) or ten (decile) backtested portfolios of increasing ex-ante volatility. Figure~\ref{fig:risk_return_100y} illustrates the advantageous risk-return profile of the low-volatility portfolio compared to the other volatility quintiles and the overall market computed on nearly a century of U.S. stock market data\footnote{ The investment universe consists of the 1000 largest stocks listed on the three major U.S. exchanges (NYSE, AMEX, and NASDAQ) from 1929 to 2020. Five quintile portfolios, each consisting of 200 stocks, were formed by sorting the investment universe by 36-month historical volatility. The portfolios are rebalanced quarterly and the positions within each quintile are equally weighted (weighting of 0.5\% per share). The data set is provided by Pim van Vliet and Jan de Koning on their private website https://www.paradoxinvesting.com/data/.}.

\begin{figure}[t] 
\centering
\includegraphics[width=0.8\linewidth]{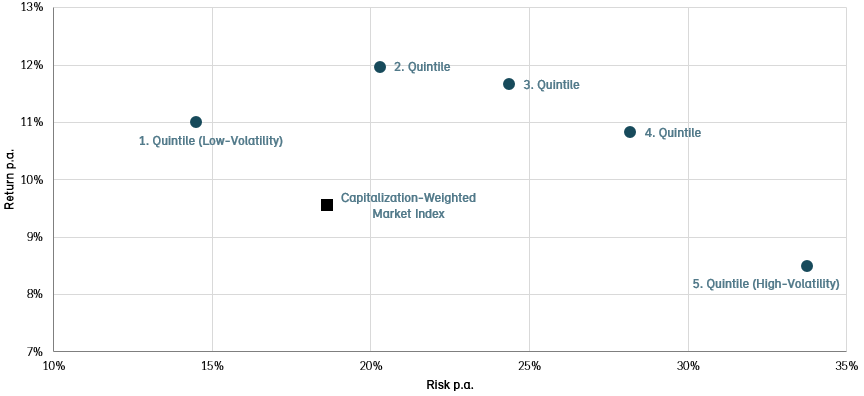}
\caption{Risk-return profile of volatility quintile portfolios and the overall market portfolio for the U.S. stock market over nearly a century.}
 \label{fig:risk_return_100y}
\end{figure}

The performance difference between the lowest and highest volatility stocks is called the low-volatility premia. It can be exploited, in principle, by forming long-short portfolios having positive weights (i.e., a long position) in the low-volatility stocks and negative weights (i.e., a short position meaning that one sells stocks) in the high-volatility stocks. In practice, however, many investors and investment vehicles like mutual funds are prohibited by regulation to short sell assets. Thus, their potential benefit from the anomalous risk-return relation is not to bet on the underperformance of high-volatility stocks, but to overweight low-volatility stocks while avoiding exposure to high-volatility titles.

Alternatively to sorting, the minimum-variance (MV) portfolio provides another solution for a low-risk portfolio by means of optimization. A parallel stream of literature investigating the empirical performance of MV portfolios aligns with the findings on low-risk portfolios based on sorts of single-stock risk measures thus further strengthening the existence of a low-risk anomaly (see \cite{bib:HaugenBaker1991} \cite{bib:ClarkedeSilvaThorley2006}, \cite{bib:ClarkedeSilvaThorley2011}).

Contrary to the pure volatility sorting-based portfolio construction method, the minimum variance procedure includes information about correlations in the portfolio selection process. As a result, the minimum variance portfolio is, although related, different from the sorting-based low-volatility portfolio for it might very well include medium- or even high-volatility stocks as long as they contribute to an overall decrease in portfolio volatility through low correlations. 
%Recall that the variance of a portfolio composed of $n$ assets with allocation vector $x = (x_1, ..., x_n)$ is given by $\sigma_P &= x^{\top} \Sigma x = \sum_{i=1}^n \sum_{j=1}^n x_i x_j \sigma_{i,j}$ where $\Sigma$ denotes the covariance matrix with elements $\sigma_{i,j} = [\Sigma]_{i,j}$.

Vice-versa, the ex-ante variances of portfolios formed from volatility-sorted subsets of the entire investable universe do not have to be monotonically increasing even though the average ex-ante volatilities of the assets within the groups are. This is, again, because of correlations. 

Hence, the ignorance of correlations in the formation of volatility-ranked portfolios poses a drawback in the existing literature analyzing the volatility puzzle. Instead of clustering stocks according to volatility and representing the subgroups by a single portfolio, an alternative would be to pre-define certain volatility targets and to sample portfolios with exactly those ex-ante volatilities from the entire collection of investable firms. Using sampling, one not only considers correlations but further overcomes the somewhat arbitrary choice of the weighting scheme for the quantile-portfolios (equally or relative to firm size), which is known to have a large impact on performance and inference \cite{bib:PlyakhaUppalVilkov2014}.

\subsection{Connection to geometry, methodology, and computational challenges}

Finance and economic textbooks usually represent the set of feasible portfolios in a risk-return space as in the left plot of Figure~\ref{fig:iso_variance}. The blue parabolic curve, whose concave part is termed the efficient frontier, highlights the boundary of the feasible set (grey area), and the vertical lines correspond to iso-variance portfolios (color-coded by increasing volatility from green to red). 
Despite the fact that we adopt this economic representation to analyze the results of our empirical study of the low-volatility anomaly, our approach consists of representing the set of volatility constrained portfolios with geometrical objects. In particular, we represent the set of long-only portfolios, i.e., portfolios with positive or zero weights, with the canonical simplex. Then, given the covariance matrix of assets' returns, the set of portfolios with a fixed level of volatility is the intersection between the canonical simplex and the boundary of an ellipsoid centered at the origin and whose shape and orientation are determined by the eigenvectors and eigenvalues of that matrix. 
The right plot of Figure~\ref{fig:iso_variance} illustrates the situation for three hypothetical assets and five variance levels (corresponding to the $y-$coordinates of the colored vertical lines in the left plot of Figure ~\ref{fig:iso_variance}), giving rise to five iso-variance ellipsoids intersecting the simplex (grey dotted surface). 
The parabolic efficient frontier in the left plot of Figure~\ref{fig:iso_variance} transforms to a piecewise-linear line in blue in the right plot of Figure~\ref{fig:iso_variance}. In either representations, it highlights the set of mean-variance Pareto-efficient portfolios. Harry Markowitz, the father of modern portfolio theory, called it the critical line\footnote{Note that Markowitz chose a geometric exposition of his ideas.} in his seminal paper \cite{bib:Markowitz1952}.

\begin{figure}[t] 
\centering
\includegraphics[width=1\linewidth]{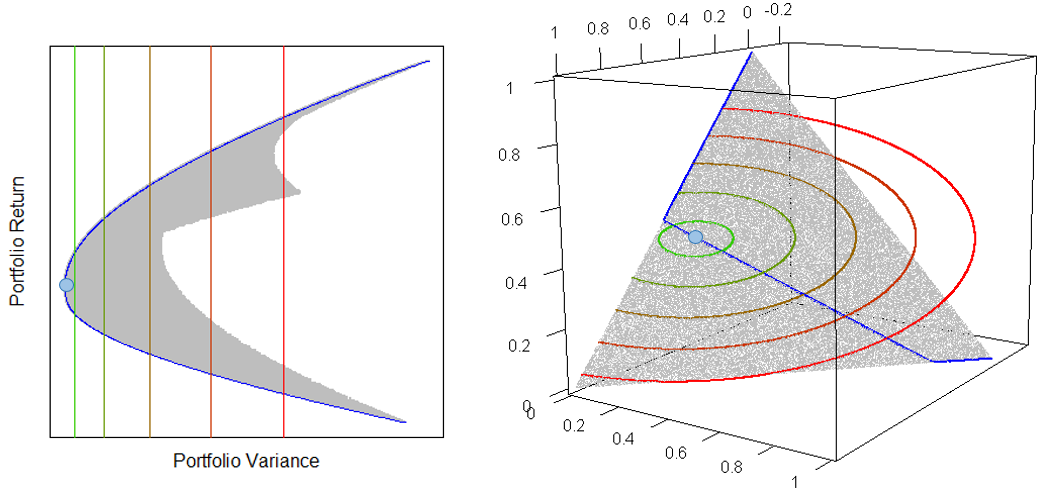}
\caption{Left plot: Feasible set and iso-volatility lines in risk-return space. The grey area represents the set of feasible portfolios. Iso-volatility level sets form vertical lines and are color-coded by increasing variance from green to red. The blue curve highlights the set of Pareto-efficient portfolios in the trade-off between risk and return. The part above the minimum variance portfolio (light blue dot) is called the efficient frontier.\\
Right plot: Feasible set and iso-volatility curves in asset weights space. The grey area depicts the simplex of feasible portfolios. Iso-volatility level sets form ellipsoidal curves, again color-coded by increasing variance from green to red. The blue line, called the critical line, highlights the set of Pareto-efficient portfolios in the trade-off between risk and return. Note that the minimum variance portfolio, again illustrated with a light blue dot, forms the centroid of the $2$-dimensional iso-volatility ellipses and indicates the point where the $3$-dimensional ellipsoid centred at the origin touches the simplex.}
\label{fig:iso_variance}
\end{figure}

%\begin{figure}[t] 
%\centering
%\includegraphics[width=0.7\linewidth]{./figures/efficient_frontier.PNG}
%\caption{Feasible set and iso-volatility lines in risk-return space. The grey area represents the set of feasible portfolios. Iso-volatility level sets form vertical lines and are color-coded by increasing variance from green to red. The blue curve highlights the set of Pareto-efficient portfolios in the trade-off between risk and return. The part above the minimum variance portfolio (green dot) is called the efficient frontier.}
%\label{fig:efficient_frontier}
%\end{figure}

%\begin{figure}[t] 
%\centering
%\includegraphics[width=0.7\linewidth]{./figures/critical_line.PNG}
%\caption{Feasible set and iso-volatility curves in asset weights space. The grey area depicts the simplex of feasible portfolios. Iso-volatility level sets form ellipsoidal curves color-coded by increasing variance from green to red. The blue line, called the critical line, highlights the set of Pareto-efficient portfolios in the trade-off between risk and return.}
%\label{fig:critical_line}
%\end{figure}

In the forthcoming empirical analysis, we estimate the covariance matrix from historical data and set several volatility levels that define a sequence of concentric ellipsoids intersecting the simplex. By sampling independently and uniformly from each intersection, we obtain sets of volatility-constrained portfolios. By computing the future returns for each sample we then capture the dependency between portfolio volatility and future portfolio return using several statistical tools from quantitative analysis. Note that one could directly extend this approach to capture the dependency between volatility and other portfolio scores.

From a geometric point of view, the intersection between the canonical simplex and the boundary of an ellipsoid in $\R^d$ is a $(d-1)$-dimensional (geodesically) non-convex and non-connected body (see Figure~\ref{fig:components_examples}), thus, forming a very challenging problem. 
%Sampling from low-dimensional manifolds is a well studied  problem with various applications e.g.\ in machine learning~\cite{Reisinger10, Welling11} (see also below).
 An additional challenge, in our case, comparing to existing work on manifold sampling, comes from non-connected-ness i.e.\ we have to estimate the volume of each connected part to achieve uniform sampling. 
 Our approach relies on Markov Chain Monte Carlo (MCMC) sampling and efficient practical Multiphase Monte Carlo (MMC) schemes for volume approximation. 

Considering previous approaches, quantile or MV portfolios correspond to points that belong to the boundary of the set of constrained volatility portfolios. In particular, they lie on the intersection between several facets of the simplex and the boundary of the ellipsoid. Consequently, the region around those portfolios is very unlikely to be sampled, i.e., be visited by the random walk. Thus, it is likely that the analysis of the volatility puzzle with those portfolios to be carried out using outliers in our statistical framework.

High-dimensional sampling from multivariate distributions with MCMC algorithms is a fundamental problem with many applications in science and engineering~\cite{Iyengar88,Somerville98,Genz09,Schellenberger09,venzke2019}. %
In particular, multivariate integration over a convex set and volume approximation of such sets ---a special case of integration--- have accumulated a broad amount of effort over the last decades. 
Nevertheless, these problems are computationally hard 
if we want to solve exactly in general dimension~\cite{Dyer88}.%\todo{vis: say sth about complexity of more general problems such as compute the volume of a ball intersected by a polytope?}. 

MCMC sampling algorithms have made remarkable progress in solving efficiently the problems of sampling and volume approximation of full-dimensional convex bodies in $\R^d$ while enjoying great theoretical guarantees~\cite{Chen18,Lee18,Mangoubi19}. 
Sampling from the boundary of a convex body has also been studied~\cite{dieker2014stochastic}. However, these algorithms could not be used in our setting since they are focused either on full-dimensional convex bodies or on sampling from the entire boundary of a convex body instead of a part of it (as in our case).

Manifold sampling is a well-studied problem~\cite{Narayanan06,Diaconis13} with various applications, e.g., in machine learning~\cite{Byrne13}.  Of special interest is the case where the manifold is a hyperpshere~\cite{Davidson18,grattarola19,reisinger10}. Moreover, sampling efficiently on constraint manifolds is a core problem in robotics~\cite{ortiz2021}. 
%\todo[inline]{Add the papers of Narayanan for sampling from hypersurfaces and some of its citations.}
Finally, in~\cite{Cong17} they propose a method for sampling multivariate normal distributions truncated on the intersection of a set of hyperplanes. 

From a practical perspective, theoretical sampling algorithms cannot be applied efficiently for real-life computations. For example, the asymptotic analysis by~\cite{Lovasz06} hides some large constants in the complexity, and in~\cite{Lee18} the step of the random walk used for sampling is too small to be an efficient choice in practice.   

Recently, practical volume algorithms have been designed by relaxing the theoretical guarantees and applying new algorithmic and statistical techniques;
they are very efficient in practice and they also guarantee high accuracy results~\cite{Emiris14,Cousins15,CoolBod}. 
Volume computation and uniform sampling have been shown to have useful applications in finance for crises detection~\cite{Cales18} and efficient portfolio allocation and analysis~\cite{PST04, HHPS2002}.
In~\cite{Karthekeyan10} they propose a polynomial-time algorithm for the more general problem of sampling and volume computation of star-shaped bodies, an important non-convex generalization of convex bodies.

In~\cite{Abbasi17} they prove that Hit-and-Run mixes fast in a more general setting that includes star-shaped bodies and spiral bodies appearing in motion planning. 
%\todo{say something about Lasserre`s paper and the limited truncated non convex implementations }
For the more general problem of approximating the volume of basic semi-algebraic sets (that is a set defined by a disjunction of polynomial equalities and inequalities) based on the so-called Moment-SOS hierarchy
we refer the reader to \cite{tacchi2020stokes,tacchi2022exploiting}.
%\todo{cite volume of union of balls paper https://hal.inria.fr/inria-00409374/document}

\subsection{Contributions}

We contribute to the literature on the low-volatility anomaly via a novel methodology to generate risk-sorted portfolios in high dimensions. Our approach directly samples uniformly distributed long-only portfolios having a certain level of volatility, which considers correlations and provides insights into the distribution of portfolios' performance statistics. Our empirical application of the geometric approach shows how the distributions of portfolios' performance statistics vary with the ex-ante volatility level and where the performance statistics of the standard sorting-based approach reside: are they close to the modes or outlying? Given that we find sorting-based results which are strong outliers with respect to (w.r.t.) the distribution of sampling-based statistics, we conclude that the classical sorting-based portfolios form bad proxies and one should be careful to base inference on them. 

On the technical side, our contributions consist of the geometric modeling of the financial problem and the construction of efficient randomized geometric tools. First, we apply proper linear transformations to end up operating on the intersection of the unit sphere with the interior of an arbitrary full-dimensional simplex. Then, we sample at the intersection and apply the inverse transformation to obtain volatility-constrained portfolios. We develop two new geometric random walks to sample from such spherical, geodesically, non-convex, and non-connected patches according to any given probability distribution. We also design a new MMC scheme to estimate the volume of a spherical patch. That is a practical randomized method based on simulated annealing and sampling from the Von-Mises Fischer distribution. Our MMC scheme generalizes and extends existing randomized volume approximation schemes~\cite{Cousins14}.
Last but not least, we offer an efficient open-source implementation in {\tt C++} with interface in {\tt R} (see Section~\ref{section:application} for details).

\paragraph*{Structure of the paper}
%\todo[inline]{Recheck at the very end}
%The remainder of the present article is structured as follows. 
%We conclude this section with basic notations that we use throwout the paper. 
%\todo{Check this at the end}
The next section then describes the geometric modeling of the financial problem in detail as well as the randomized geometric methods that we build, such as MCMC sampling and volume approximation.
The section~\ref{section:application} contains details of our implementation and details about our algorithmic pipeline we introduce to analyze the volatility puzzle. The section~\ref{sec:results} presents the results and the analysis of the volatility anomaly in European and U.S.\ stock markets using our methods. 
Last, in section~\ref{section:discussion} we discuss the possible future work to improve the efficiency of our methods and other financial problems that our framework could be used to address.

\section{Geometric modeling and algorithms}\label{sec:algorithms_modeling}

Our approach consists of the geometric modeling of long-only volatility-constrained portfolios. Within this framework, we develop two sampling algorithms and an MMC scheme for volume approximation. We first present the geometric framework and certain mathematical tools before presenting our algorithms. 
%\todo{An overview of the algorithms is probably missing from here. i.e.\ How exactly the two walks and the volume estimation are used in the modelled problem to sample uniformly. Maybe this discussion should be placed at the end of 2.1}

\paragraph*{Notations}

We denote a full dimensional (convex) body with a capital letter and if applicable with an index we denote its dimension; for example, the unit ball in $\R^d$ is  $B^d$. We denote the lower dimensional, usually (geodesically) non-convex, bodies with a calligraphic letter; for example $\mathcal{S}_{d-1}$  is the $(d-1)$ dimensional sphere in $\R^d$. They are usually parts of the surface of a full-dimensional convex body.

\subsection{Geometric modeling}\label{sec:geometric_modeling}
%Finance Formulation / Model: \\
%- $d$: number of assets \\
%- $d+1$ constraints that define a simplex $S\subseteq\mathbb{R}$\\ 
%- sequence of $k$ co-centric ellipsoids $E_i: x\Sigma x^T\leq c_i$ for $i\in[1,\dots,k]$ and $c_1>c_2>\cdots>c_k$

%Problem:
In finance, a portfolio is a collection of assets. 
Each portfolio allocates a percentage of a given budget to every asset. 
Thus, in our setting, the set of long-only portfolios, is the canonical $d$-dimensional simplex 
\begin{equation}
\Delta_d := \{ x\in\RR^{d+1}\ |\ x_i \geq 0,\ \sum_{i=1}^{d+1} x_i = 1 \} \subset \RR^{d+1},
\end{equation}
where each point represents a portfolio and $d+1$ is the number of assets. The vertices represent portfolios composed entirely of a single asset.
The portfolio weights, i.e., the fractions of investment for each asset, are non-negative and sum up to 1. Notice that $\Delta_d$ is a $d$-dimensional body that lies in $\RR^{d+1}$; that is, it is a {\em lower dimensional} body. 

\begin{figure}[!t] 
\centering
\includegraphics[width=0.32\linewidth]{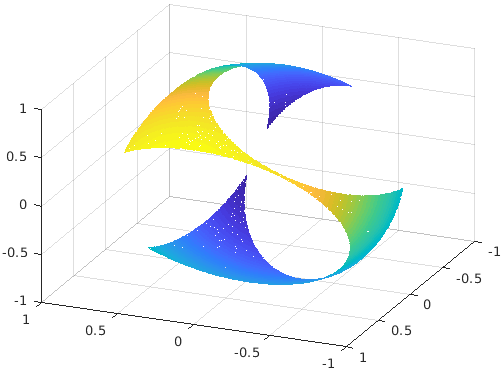}
\includegraphics[width=0.32\linewidth]{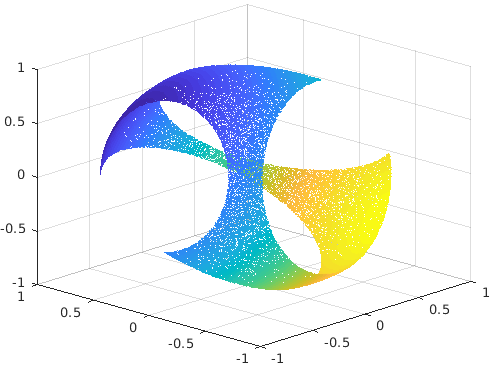}
\includegraphics[width=0.32\linewidth]{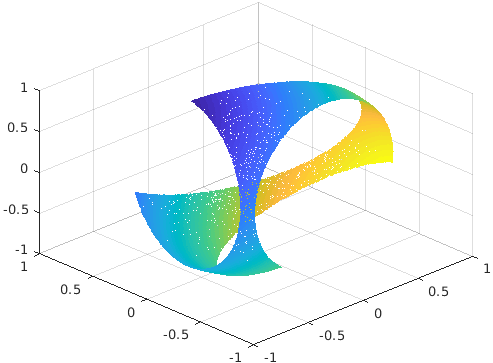}\\
\vspace{0.4cm}
\includegraphics[width=0.32\linewidth]{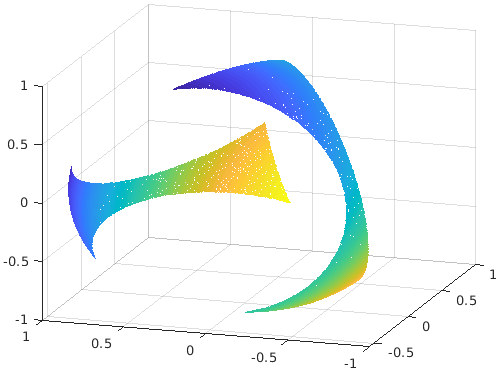}
\includegraphics[width=0.32\linewidth]{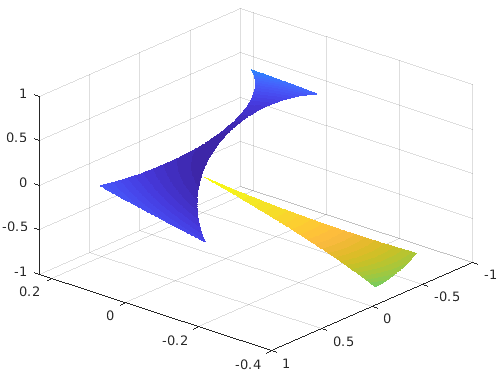}
\includegraphics[width=0.32\linewidth]{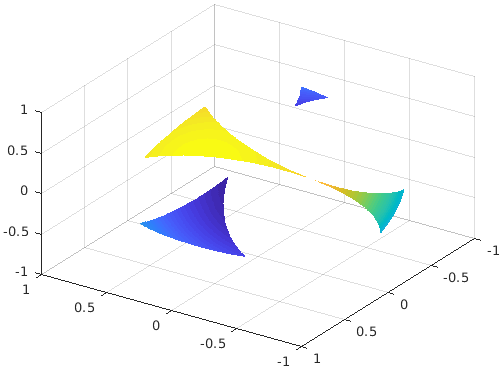}
\caption{Six examples in $\R^3$ of the body $K = \mathcal{S}_2 \cap \Delta$ (the unit sphere intersected by the interior of a simplex) that we sample from. In general, $K$ is a non-connected and geodesically non-convex body in $\R^d$.}
\label{fig:components_examples}
\end{figure}

Given a vector of assets' returns $R\in\RR^{d+1}$, and a portfolio $x\in\Delta_d$, we say that  the return of $x$ is $f_{ret}(x, R) = R^T x$. Similarly, given the positive definite covariance matrix $\Sigma\in\R^{(d+1)\times (d+1)}$ of the distribution of the assets' returns, we define the portfolio volatility as $f_{vol}(x, \Sigma) = x^T\Sigma x$. Thus, to model portfolios' volatility we employ ellipsoids intersecting $\Delta_d$. 
The set,
\begin{equation}
  \mathcal{E}_d^{c} \cap \Delta_d =  \{ x\in\Delta_d\ |\ x^T \Sigma x = c,\ c\in\R_+ \} \subseteq \R^{d+1} ,
\end{equation}
corresponds to the portfolios with volatility $c$. Notice that portfolios, that is the the points, that belong to the (centered at the origin) ellipsoid, $\mathcal{E}_d^{c} := \{ x\in\R^{d + 1}\ |\ x^T \Sigma x = c,\ c\in\R_+ \} \subseteq \R^{d+1}$, achieve volatility equal to $c$. The portfolios in the interior of $\mathcal{E}_d^{c}$ achieve lower volatility, while those in the complement of $\mathcal{E}_d^{c}$ achieve higher volatility than $c$.
%\todo{vis: It should be useful here to explain, describe how the previous modeling with limitations (sorting) would be represented geometrically}
%We want to sample from $S_i:\ S\cap E_i^o\cap E_{i+1}^c$ for $i\in [1,\dots,k]$, where $E^c,E^o$ the complement and the interior of an ellipsoid $E$. 

%Our approach is to capture the dependency between portfolios' return and volatility during a time period in a stock market. When this dependency is (almost) negative we declare this period as one of volatility anomaly. To achieve this objective we generate uniformly distributed portfolios for several values of volatilities and then we evaluate the future returns of the generated portfolios. Consequently, the latter evaluations allow us to estimate the distribution of portfolios' return for each volatility target. Then, with the simple statistical analysis, we detect volatility anomalies. \footnote{Tolis: Improve, this paragraph. Our pipeline is not clear to me.}
%Portfolio performance refers to evaluating the performance of a given portfolio. It is essentially a process of comparing a portfolio’s return with the return earned on a benchmark portfolio (or one or more other portfolios or indices).
%To study the performance of the portfolios with volatility level $c$ we wish to generate uniformly distributed portfolios from $\mathcal{E}_d \cap \Delta_{d}$.

Geodesics on ellipsoids (as distinct from spheres) are not, in general, closed. 
That is, while we can compute geodesics on a sphere by exploiting  spherical trigonometry, this is not the case for ellipsoids where one has instead to solve differential equations~\cite{Karney12}. Thus, for efficiency, we map the ellipsoid $\mathcal{E}_d^{c}$ onto the unit hypersphere $\mathcal{S}_{d-1} \subset \R^d$ and we apply the same transformation to the simplex $\Delta_{d}$ to obtain a full-dimensional simplex $\Delta\in\R^d$. 
In particular, first, we use an orthonormal basis that spans the linear subspace that $\Delta_d$ lies on to obtain both an ellipsoid and a simplex in $\R^d$. 
Then, we transform the latter ellipsoid into $\mathcal{S}_{d-1}$ and use this transformation to obtain a simplex $\Delta\in\R^d$.
Therefore, instead of sampling from $\mathcal{E}_d^{c} \cap \Delta_d$,
we sample from 
\begin{equation}\label{eq:body_we_sample}
	\cK := \mathcal{S}_{d-1} \cap \Delta .
\end{equation}
We can use the inverse transformations obtain uniformly distributed portfolios with volatility $c$ since the transformation is isometric, see e.g.,~\cite{chalkis21}.

In general, $\cK$ is a set of geodesically non-convex\footnote{A subset $C$ of a surface is said to be a geodesically convex set if, given any two points in $C$, there is a unique minimizing geodesic contained within $C$ that joins those two points.
A geodesic is a curve representing in some sense the shortest path or arc between two points on a surface.}
spherical surface patches (see Figure~\ref{fig:components_examples} for a few examples).
We call these patches the (connected) components of $\cK$; we denote them by $\cK_i,\ i\in[M]$, where $M$ is their cardinality. 
To sample uniformly from $\cK$, we sample uniformly from each component
according to its relative volume. 
In particular, first we sample $u\sim\mathcal{U}(0,1)$. 
Afterwards, if $u\in [ \sum_{i=1}^mw_i, \sum_{i=1}^{m+1}w_i ]$, for some $m < M$, 
then we sample a uniformly distributed point from $\cK_m$, where $w_i =  {\vol(\cK_i)} /{\vol(\cK)},\ i\in[M]$ is the relative volume of the $i$-th component.

%For each time period we sample portfolios for several volatility levels to analyze the performance %for each level and to detect possible anomalies.
%Section~\ref{section:application} presents typical values for the number of volatility levels that %appear in practical performance analysis. 

%\subsection{The geometry of the problem and the connected components}

To identify and represent the components of $\cK$ we use the vertices and the edges of $\Delta$. 
In particular, we consider the $1$-skeleton of $\Delta$ i.e.\ the graph whose vertices are the vertices of $\Delta$, with two vertices adjacent if they form the endpoints of an edge of $\Delta$.
Note that in the case of $\Delta$ the $1$-skeleton is a clique. 

We identify the edges of $\Delta$ that $\mathcal{S}_{d-1}$ intersects and remove them from the $1$-skeleton. 
We also identify the vertices of $\Delta$ that lie in the interior of $\mathcal{S}_{d-1}$ and remove them as well as their adjacent edges.
We denote the resulting graph by $G$.
There is a bijection between the connected components of $G$ and the connected components of $\cK$. 
Thus, we represent each connected component of $\cK$ using the set of vertices of the corresponding connected component of $G$. To decide if a given point $p\in\Sd$ belongs to a certain component of $\cK$ we develop a membership oracle. One call costs $O(d^2)$ operations (see Appendix~\ref{appnd:mem_oracle}).

%\subsection{Computing a point in a connected component}

\subsection{Sampling from a connected component of $\cK$}

We introduce two geometric random walks, namely the Great Cycle Walk (GCW) and the Reflective Great Cycle Walk (ReGCW), to sample from a (connected) component of $\cK$. To design these algorithms we employ the geodesics of $\Sd$. 
Note that the great circles on the $\Sd$ are the intersection of the $\Sd$ with $2$-dimensional hyperplanes that pass through the origin in the $\R^d$.
These great circles are the geodesics of the $\Sd$~\cite{Byrne13}.

\subsubsection{Great Cycle Walk (GCW)} 
\label{sec:GCW}

Great Cycle Walk (GCW) is a random walk to sample from any probability density function $\pi(x)$ supported on a connected component of $\cK$.
GCW generalizes the Hit-and-Run sampler~\cite{Smith93} on a spherical patch.
At each step, GCW starts from a point $p\in \cK$ and picks uniformly a great cycle $\ell$ of $S_{d-1}$ passing through $p$. Then, it computes the part of the great cycle that lies in $\cK$ and contains $p$. It samples a point from that part of $\ell$ according to $\pi_{\ell}$ to set the next Markov point, where $\pi_{\ell}$ is the restriction of $\pi$ on $\ell$. 
We prove (Theorem~\ref{thm:GCW}) that the (unique) stationary distribution of this algorithm is $\pi$ and moreover, that GCW converges to $\pi$ from any starting point in $\cK$. 

\begin{algorithm}[t]
  \caption{GCW$(\Delta, p, \pi)$\label{alg:GCW}}
	\label{alg:cgw}
%	\Set\cKwProg{BO}{Function}{}{end}
	\SetKwInOut{Input}{Input}
	\SetKwInOut{Output}{Output}
    \SetKwInOut{Require}{Require}
    \SetKwRepeat{Do}{do}{while}

    \Input{Simplex $\Delta$; point $p$; PDF $\pi$.}

    \Require {$\mathcal{S}_{d-1} \cap \Delta \neq \emptyset$; point $p\in \cK=\mathcal{S}_{d-1} \cap \Delta$}

	\Output{Next Markov point in $\cK$}

    \BlankLine
    
    Pick a uniform vector $v$ from $\mathcal{S}_{d-1} \cap \{ x\in\R^d\ |\ p^Tx = 0 \}$\;
    Let the great cycle $\ell(\theta) := \{p\cos\theta + v\sin\theta,\ \theta\in[0,2\pi] \}$\;
    Let $(\theta^-,\theta^+)$ the values s.t.\ $\ell_p(\theta) := \{ p\cos\theta + v\sin\theta,\ \theta\in[\theta^-,\theta^+] \}$ is the part of $\ell(\theta) \cap \Delta$ that contains $p$\;
    Pick $\tilde{\theta}$ from $\ell_p(\theta)$ according to $\pi_{\ell}$\;
    \KwRet $p\cos\tilde{\theta} + v\sin\tilde{\theta}$\;
\end{algorithm}

%\todo{3 and 4 lines in CGW pseudocode seem wrong and have undefined quantities}

Some details of the computations are in order. 
GCW (at each step) chooses uniformly a great cycle passing from a point $p$ by sampling uniformly at random a unit vector $v$ from the hypersphere $S_{d-1}$ restricted to the hyperplane $\mathcal{H}_p := \{ x\in \R^d\ |\ p^Tx = 0 \}$. 
The parametric equation of the great cycle is
\begin{equation}\label{eq:gc_parametric_equation}
\ell(\theta) := \{p \cos \theta  + v \sin \theta ,\ \theta \in [-\pi, \pi]\} .
\end{equation}
We compute $v$ by sampling uniformly at random a point $u$ in $\mathcal{S}_{d-1} \cap \mathcal{H}_p$
as follows 
\begin{equation}\label{eq:projection}
v =  \frac {(I_d - pp^T)u} {\| (I_d - pp^T)u \|_2},
\end{equation}
where $I_d$ is the $d \times d$ identity matrix.
In this way, $v$ is the normalized projection of $u$ on the hyperplane $\mathcal{H}_p$ and it is uniformly distributed in $\mathcal{S}_{d-1}$.

GCW computes the connected part of $\ell(\theta)$ that lies in $\Delta$ and contains $p$ by computing the intersection of $\ell(\theta)$ with each facet of $\Delta$.
%and identifying the connected part of $\ell(\theta) \cap P$ that contains $p$. 
For this, it computes the smallest positive and the largest negative solution of the following equations,
\begin{equation}\label{eq:bo_cgw_1}
a_j^T \ell(\theta) = b_j \iff a_j^Tx_i \cos\theta + a_j^Tv_i \sin\theta = b_i,\ \theta\in[-\pi, \pi],\ j\in[d+1] ,
\end{equation}
where $a_j\in\R^d$ are normal vectors of the facets of $\Delta$.
Let $z_j = {b_j} /{\sqrt{(a_j^Tx_i)^2 + (a_j^Tv_i)^2}}$. 
If $z_j\in[-1,1]$, then
the values of $\theta$ that correspond to the intersections are
\begin{equation}\label{eq:bo_cgw_2}
\theta^-_j,\ \theta^+_j = \pm\cos^{-1}(z_j) + \tan^{-1}\bigg( \frac {a_j^Tv_i} {a_j^Tp} \bigg) .
\end{equation}
Otherwise, the great cycle does not intersect with the $j$-th facet. 
Thus, GCW keeps $\theta^+ = \min\limits_{j\leq m}\{ \theta^+_j \}$ and $\theta^- = \max\limits_{j\leq m}\{ \theta^-_j \}$ and the intersection points are $\ell(\theta^+)$ and $\ell(\theta^-)$, respectively. 

\begin{thm}\label{thm:GCW}
The stationary distribution of Alg.~\ref{alg:GCW}, CGW($\Delta,p,\pi$), where the starting point $p$ belongs to $K$, is $\pi$ for any starting point $p$ in $K$.
\end{thm}

Finally, $\Delta$ has $d$ facets while the computations of the intersection of $\ell(\theta)$ with a facet of $\Delta$ takes $O(d)$ arithmetic operations, which lead to the following remark.

\begin{rem}
The cost per step of GCW is $O(d^2)$ arithmetic operations.
\end{rem}

GCW can also be used to sample uniformly from a component of $\cK$ (see Appendix~\ref{appnd:gcw_uniform}). However, the next geometric random we introduce, namely Reflective Great Cycle Walk, has superior practical performance compared to GCW for uniform sampling.

\subsubsection{Reflective Great Cycle Walk (ReGCW)}\label{sec:regcw}

\begin{figure}[t] \centering
\includegraphics[width=0.7\linewidth]{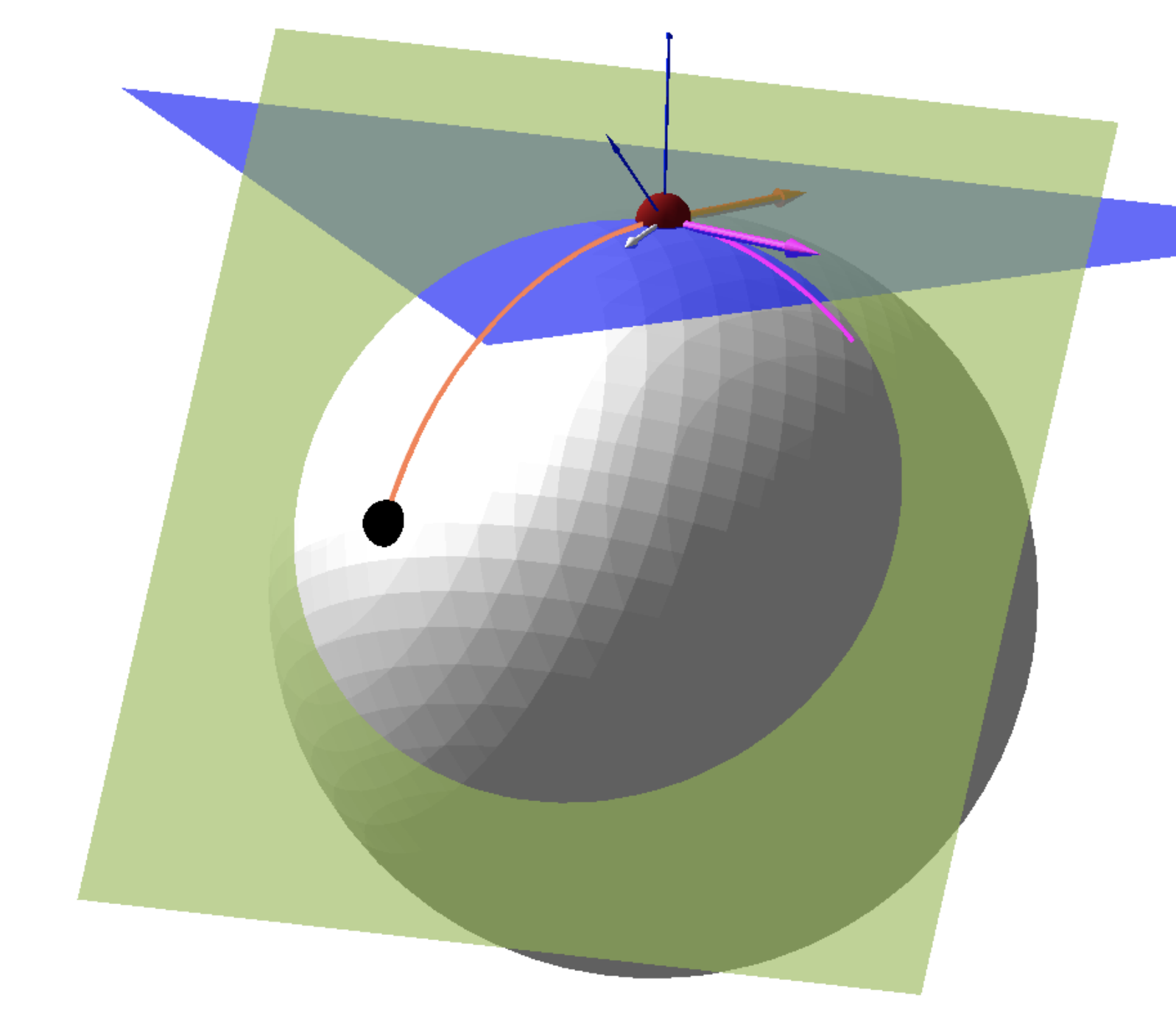}
\caption{An illustration of the reflection of the ReGCW.
The orange trajectory, starting from the black point, hits the boundary of the component defined by the intersection of the sphere with the green hyperplane (red point).
The blue hyperplane is the tangent space at the intersection point. 
The resulted trajectory is  in magenta. 
\label{fig:ReGCW}}
\end{figure}

We introduce the Reflective Great Cycle Walk (ReGCW), a geometric random walk that operates on a connected component of $\cK$ and converges to the uniform distribution. ReGCW is a generalization of Billiard Walk~\cite{Gryazina14} on a spherical patch. 
Similar to GCW it starts from a point $p$ in a connected component $\cK_i$. 
At each step, it generates uniformly a great cycle passing through $p$
and it computes a trajectory length $L=-\tau \ln \eta$, where $\eta$ is
a uniform number in $[0,1]$, i.e., $\eta\sim\mathcal{U}[0,1]$, and $\tau$ is a predefined constant.
When the generated great cycle hits $\partial \Delta$, it reflects and the reflected curve is also a (part of a) great cycle of $\Sd$. 
ReGCW returns the next Markov point when it travels distance $L$; if the number of reflections exceeds a given upper bound $\rho\in\mathcal{N}_+$ then, the next point is $p$ itself. 
It is useful to set a bound on the number of reflections to avoid computationally hard cases where the trajectory may stick in corners. 
We detail our choices for $\tau$ and $\rho$ in Section~\ref{section:application}.

\begin{algorithm}[t]
  \caption{Reflective Great Cycle Walk$(\Delta, p, \rho, \tau)$}
	\label{alg:regcw}
%	\SetKwProg{BO}{Function}{}{end}
	\SetKwInOut{Input}{Input}
	\SetKwInOut{Output}{Output}
    \SetKwInOut{Require}{Require}
    \SetKwRepeat{Do}{do}{while}

    \Input{Simplex $\Delta\in\R^d$; current Markov point $p$; upper bound on the number of reflections $\rho$; length of trajectory parameter $\tau$;}

    \Require {$\mathcal{S}_{d-1} \cap \Delta \neq \emptyset$; point $p\in \mathcal{S}_{d-1} \cap \Delta$}

	\Output{Next Markov point in the same component of $\mathcal{S}_{d-1} \cap \Delta$ as $p$}

    \BlankLine
    %\For {$j=1,\dots ,W$} {
       $L \leftarrow -\tau\ln\eta$, $\eta\sim \mathcal{U}(0,1)$ \tcp{length of the trajectory}

       $i\leftarrow 0$ \tcp{current number of reflections}

       $p_0\leftarrow p$ \tcp{initial point of the step}

       Let the hyperplane $\mathcal{H}_p := \{ x\in\R^d\ |\ p^Tx = 0 \}$\;
       Pick a uniform vector $v$ from $\mathcal{S}_{d-1} \cap \mathcal{H}_p$\;

    	% length of the trajectory $L \leftarrow -\tau\ln\eta$, $\eta\sim \mathcal{U}(0,1)$\;
    	% current number of reflections $i\leftarrow 0$\;
    	% initial point of the step $p_0\leftarrow p$\;
    	% uniformly distributed direction $\varv_0$\;

    	\Do{$i\leq \rho$} {
    		Let the curve $\ell(\theta) := \{ p\cos\theta + v\sin\theta,\ \theta\in[0,L] \}$\;
    		$\tilde \theta \leftarrow \argmin\limits_{\theta\in[0,L]}\{ \ell(\theta)\in\partial \Delta \}$ ; \tcp{intersection angle}
    		\lIf{$L < \tilde\theta$ or $\partial \Delta \cap \ell(\theta) = \emptyset$} {
    			\KwRet $p\cos L + v\sin L$ 
    		}
            
    		$p \leftarrow \ell(\tilde\theta)$ and $v \leftarrow \ell'(\tilde \theta)$ ; \tcp{point and direction update}
    		Let $s$ the normalized projection of inner vector of the facet of $\Delta$ at $p$ on $\mathcal{H}_p$\;
    		$v \leftarrow v - 2(v^Ts) s$ \tcp{reflected direction}
    		$L \leftarrow L - \tilde \theta$\;
    		$i \leftarrow i+1$\;
    	}
    	\leIf{$i=\rho$} {
    		\KwRet $p_0$
    	}
    	{
    		\KwRet $p$
    	}
    %}
    %$y\leftarrow x$\;

\end{algorithm}

Some details of the computations are in order.
At each step of ReGCW, we are at a point $p$, we denote by $\ell(\theta) := \{ p\cos\theta + v\sin\theta,\ \theta\in[0,L] \}$ the part of the great cycle emanating from $p$ and has length $L$; 
since we are on the unit sphere, geodesic length is numerically equal to the radian measure of the angles that the great circle arcs subtend at the center. 
We compute the smallest angle $\tilde \theta$ as in Equation~(\ref{eq:bo_cgw_1}), hence, $q := \ell(\tilde\theta)\in\partial \Delta$ is the point on the components boundary hit by the ReGCW trajectory. 
If $\tilde \theta < L$, then we compute the reflection of $\ell(\theta)$ at $\tilde \theta$ as follows: Let $a\in\R^d$ the normal vector of the facet of $P$ that $\ell(\theta)$ hits. 
We compute the normalized projection of $a$ onto the hyperplane $\mathcal{H}_q := \{ x\in \R^d\ |\ q^Tx = 0 \}$, say $a'$, using the Equation~(\ref{eq:projection}). 
Then, the reflection of $\ell(\theta)$ at $q$ is,
\begin{equation}\label{eq:reflected_curve}
    \ell_r(\theta) = \{ q\cos\theta + v_r\sin\theta,\ \theta\in [0, L-\tilde \theta] \}, 
\end{equation}
where $v_r = v - 2(v^T\alpha')\alpha'$. 
We also update the length travelled so far, i.e., $L \leftarrow L - \tilde\theta$. 
When $L < \tilde\theta$ we set $\ell(L)$ as the next Markov point. 
The defined reflection operator guarantees that $v_r\in \mathcal{H}_q$ which implies that $\ell_r(\theta)$ is a part of a great cycle.
A single reflection of the ReGCW is depicted in Figure~\ref{fig:ReGCW}.
%\todo{what are the magenta, orange, white, blue (x2) arrows?}

\begin{thm}\label{thm:gcw_convergence}
The Reflective Great Cycle Walk of Alg.~(\ref{alg:regcw}) has a unique stationary distribution, which is the uniform distribution; it converges from any starting point.
\end{thm}

Finally, ReGCW performs at most $\rho$ reflections per step, while the computation of the intersection of $\ell(\theta)$ with $\partial \Delta$ costs $O(d^2)$ arithmetic operations, which leads to the following remark.

\begin{rem}
The cost per step of ReGCW is $O(\rho d^2)$ arithmetic operations.
\end{rem}

\subsection{Practical volume approximation}\label{sec:volume}

At a high level, our method is based on that of~\cite{Lovasz06b, Cousins14} and on the practical variant~\cite{Cousins15}. However, those algorithms are designed for full dimensional bodies in $\R^d$. Thus, we have to make the necessary practical adjustments to develop a practical volume estimation method. Since in our case we estimate the volume of geodesically non-convex spherical patches 
we do not have most of the theoretical guarantees appeared in previous work~\cite{Lovasz06b, Cousins14}. %\todo{Clearly state the differences if any} 
In particular, given a connected component $\cK_i,\ i\in[M]$, for any sequence of $k$ functions $f_j:\mathcal{S}_{d-1} \rightarrow \R_+,\ j\in[k]$ we consider the following representation of $\vol(\cK_i)$,
\begin{align}
\label{eq:integral_annealing}
    \vol(\cK_i) & = \int_{\cK_i} f_k dx \, \frac{\int_{\cK_i} f_{k-1} dx}{\int_{\cK_i} f_k dx} \frac{\int_{\cK_i} f_{k-2} dx}{\int_{\cK_i} f_{k-1} dx} \cdots \frac{\int_{\cK_i} dx}{\int_{\cK_i} f_1 dx}  \\
    & = \Bigg( \frac{1}{\int_{\cK_i} f_k dx} \frac{\int_{\cK_i} f_{k} dx}{\int_{\cK_i} f_{k-1} dx} \cdots \frac{\int_{\cK_i} f_1 dx}{\int_{\cK_i} dx} \Bigg)^{-1},
    \quad \text{ for } i\in[M] ,
\end{align}
while we prefer the right hand of the Equation~(\ref{eq:integral_annealing}) for reasons we explain in the sequel. We set each $f_j$ to be proportional to the von-Mises Fischer (vMF) distribution, i.e.,
\begin{equation}
    f_j(x) = e^{a_j(\mu^Tx)},\ x, \mu\in \mathcal{S}_{d-1} ,\ a_j >0,
\end{equation}
where $\mu$ is the mean and $a_j$ is the inverse of the variance. The vMF distribution is  the restriction of the spherical Gaussian distribution on the hypersphere $\mathcal{S}_{d-1}$~\cite{Mardia75}.

In previous work~\cite{Lovasz06b, Cousins14} they use exponential or Gaussian functions in Equation~(\ref{eq:integral_annealing}). They first compute the smallest variance $\alpha_k$ s.t.\ the mass of $f_K$ is almost inside the convex body. Then, they compute the sequence of functions by cooling the variance until they reach the uniform distribution. In our case, we can not have a guarantee for $\alpha_k$, and thus, we start from the uniform distribution and then we decrease the variance until we reach a distribution with sufficiently small variance. To stop decreasing the variance, in each step of our schedule, we sample from $f_j$ and we probabilistically bound the proportion of the mass outside $\cK_i$ exploiting Bernoulli trials.

By standard error analysis~\cite{Jeter05}, to estimate $\vol(\cK_i)$ within relative error $\epsilon$ it suffices to estimate each integral ratio in Equation~(\ref{eq:integral_annealing}) within error $\epsilon_k = O(\epsilon / \sqrt{k})$ while the $\int_{\cK_i} f_kdx$ is computed within an error $\epsilon_0 < \epsilon$. 

To estimate each integral ratio within $\epsilon_k$, we rewrite it as 
\begin{equation}\label{eq:Y_rv}
    Y_j = \frac{\int_{\cK_i} f_j dx}{\int_{\cK_i} f_{j-1} dx} = \int_{\cK_i} \frac{f_j}{f_{j-1}}\frac{f_{j-1}}{\int_{\cK_i} f_{j-1}dx} .
\end{equation}
Then, we use GCW to generate $N$ random samples from a distribution proportional to $f_{j-1}$ and restricted to $\cK_i$ and we estimate each integral ratio using the following estimator
\begin{equation}\label{eq:ratio_estimation}
    R_j = \frac{1}{N} \sum_{l=1}^N \frac{f_j(x_l)}{f_{j-1}(x_l)},
\end{equation}
while $E[Y_j] = \lim\limits_{N\rightarrow \infty} R_j$.
Using Chebyshev's inequality, when $Var[Y_j] / E[Y_j]^2 \leq 1$ we guarantee that $N = \widetilde{O}(1)$ points suffice to approximate~(\ref{eq:Y_rv}) within relative error $\epsilon_k$~\cite{Cousins14}. 
%\todo{Suffices for what? I think that we need the error approximation here.}
\vspace{0.2cm}\\
\noindent
\textbf{Fixing the sequence.} In~\cite{Cousins14} they prove that when $f_j$ are Gaussian functions in $\R^d$, if $\alpha_j = \alpha_{j-1} \bigg( 1 + \frac{1}{d} \bigg)$, then $Var[Y_j] / E[Y_j]^2 \leq 1$. To the best of our knowledge it is unclear if the Lemma 3.2 in~\cite{Cousins14} can be extended in our framework. 
%However, in our case each $f_j$ is defined on $\mathcal{S}_{d-1}$\todo{vis: we should be careful here. This reads as the sphere cannot be embedded in a compact subset of $\R^d$ which is probably wrong.}, and thus, we do not have such guarantees. 
To fix the sequence of variances we define the following practical annealing schedule based on the practical techniques in~\cite{Cousins15}. We set
\begin{equation}
    \alpha_j = \alpha_{j-1} \bigg( 1 + \frac{1}{d} \bigg) ^r,\ r\in\R_+ .
\end{equation}
Then, we sample $N'$ %\todo{we should give a reason here, why quadratic? Maybe leave it as a variable here, e.g. $N'$ points and then in epxriments say that quadratic \# of points suffice for our applications} 
points with GCW from $f_{j-1}$ and we search for the maximum $r$ s.t.\ the ratio of the variance over the square of the average value of $f_j/f_{j-1}$ evaluated on this sample lie in an interval $[1-\delta,1]$ for a predefined small value of $\delta$. To estimate the desired value of $r$ we binary search in an interval $r_{\min}, r_{\max}$, where $r_{\min} = 0$ and $r_{\max}$ is found by setting $r_{\max} = 2^n$ and $n$ the smallest integer s.t.\ the average value of the ratio  $f_j/f_{j-1}$ with $\alpha_j = \alpha_{j-1} \bigg( 1 + \frac{1}{d} \bigg) ^{r_{\max}}$ is larger than $1$. We follow~\cite{Cousins15} to set $N'$ for practical computations, and thus, we choose the value $1200 + d^2$.
\vspace{0.2cm}\\
\noindent
\textbf{First variance.} To compute $\alpha_1$ we set $\alpha_0 = 0$ and we use $N'$ %\todo{similarly here} 
uniformly distributed points in $\cK_i$
generated by ReGCW. Then, we binary search for $\alpha_1$ from a proper interval s.t.\ the average ratio between $f_1$ and the uniform distribution is smaller than $1$.
\vspace{0.2cm}\\
\noindent
\textbf{Last variance.} To stop we compute a large enough $\alpha_k$ s.t.\ almost the entire mass of the distribution proportional to $f_k$ and restricted to $\cK_i$ is inside $\cK_i$, i.e., for a predefined $\epsilon_0$ we have,
\begin{equation}
    \int_{\cK_i} f_k dx = (1-\epsilon_0)\int_{\mathcal{S}_{d-1}} f_k dx 
\end{equation}
with high probability. To achieve this objective, when we compute a new $f_j$, we sample $\nu$ points from the corresponding exact vMF distribution supported on $\mathcal{S}_{d-1}$ using the algorithm in~\cite{Kurz15} and we stop when less than $\epsilon_0 \nu$ points are outside $\cK_i$. 
Clearly, from Hoeffding's inequality $\nu=O\bigg(\log\bigg(\frac{1}{1-\zeta}\bigg)\bigg/\epsilon_0^2\bigg)$ points suffices to guarantee that $\int_{\cK_i} f_k dx \geq (1-\epsilon_0)\int_{\mathcal{S}_{d-1}} f_k dx$ with probability $1-\zeta$. % (see Appendix~\ref{appnd:computing_ak}). 
For practical computations we set $\epsilon_0 = 0.05$.
\vspace{0.2cm}\\
\noindent
\textbf{Ratio convergence.} If the points generated by GCW were independent, then we would use the theoretical bound on $N$, derived from Chebyshev's inequality in~\cite{Cousins14} to estimate each integral ratio. 
However, these samples are correlated and thus, we use the same convergence criterion as in~\cite{Cousins15}. 
In particular, for each point, we generate we update the value of the integral ratio and we store the last values on a sliding window $W$. 
We declare convergence when,
\begin{equation}
    (\max(W) - \min(W)) / \min(W) \leq \epsilon_k/2,
\end{equation}
where $\max(W)$ and $\min(W)$ correspond to the maximum and minimum values of the sliding window respectively. As in~\cite{Cousins15} it is unclear how to obtain a good bound on the probability of failure with relation to the window size. To set the length of the sliding window we follow~\cite{Cousins15} and we set it equal to $4d^2 + 500$.% empirically (see section~\ref{section:application}).
%When we fix the sequence of $f_j$ in Equation~(\ref{eq:integral_annealing}) we would like to estimate each integral ratio within relative error $\epsilon_k$ using a random sample from $f_{j-1}$ and the estimator in Equation~(\ref{eq:ratio_estimation}). 
\vspace{0.2cm}\\
\noindent
\textbf{Sampling from a segment.} To sample from mVF distribution $\pi(x)\propto e^{\alpha (\mu^Tx)}$ restricted to a connected component $\cK_i$ we use GCW algorithm. In each step of the random walk we have to sample from mVF restricted to a curve $\ell(\theta) = \{ p\cos\theta + v\sin\theta, \theta\in[\theta_1, \theta_2] \}$, i.e.,
\begin{equation}
    \pi_{\ell}(\theta) \propto e^{\alpha (\mu^Tp\cos\theta + \mu^Tv\sin\theta)},\ \theta\in[\theta_1, \theta_2] .
\end{equation}
To sample from the latter univariate distribution we use Metropolis-Hastings algorithm~\cite{Chib95} (see section~\ref{section:application}).  

\section{Implementation and experiments}\label{section:application}

We perform an empirical study using previously introduced randomized geometric tools for the detection and analysis of the low-volatility anomaly on real data from global stock markets and comment on the performance of our implementation. 

\subsection{Implementation}
\label{sec:implementation}

We present the implementation of our algorithms and the tuning of
various parameters. 
We provide a complete open-source software framework to address low-volatility detection in stock markets with hundreds of assets. 
%The framework loads a set of asset returns, it estimates the covariance matrix of the asset returns and fixes $m$ values of volatilities, for a given $m\in\mathbb{N}$. 
%Then, it computes $\cK$ and samples uniformly from it.
Our code lies on a public domain\footnote{\url{https://github.com/TolisChal/volume_approximation/tree/low_volatility}}.
%All the results are reproducible using our publicly available code
%\todo{vis: I think the readme should be replaced by clear instructions of how to build and run the code and maybe an small example.} and \todo{CB: Actually, data are not public, so our results are not reproducible.} public available data (Sect.~\ref{sec:data}).
The core of our implementation is in {\tt C++} to optimize performance while the user interface is implemented in {\tt R}. The package employs
\eigen~\cite{eigenweb} for linear algebra, \boost~\cite{boostrandom} for random number generation, and expands \volesti~\cite{chalkis2020volesti}, an open-source library for high dimensional MCMC sampling and volume approximation. In our software we also use the package \logconcdead~\cite{Culejss09} to fit logconcave distributions for the analysis of the samples generated by our random walks (see section~\ref{sec:results}).
%\todo{add logconcdead ref}

\subsubsection{Parameter tuning for practical performance}
To obtain an efficient implementation for our methods we introduce the following parameterizations.
\vspace{0.1cm}\\
\noindent
\textbf{Computing a starting point on $K_i$.} We first compute the maximum inscribed ball in the intersection $\Delta \cap B_d$, where $B_d$ is the unit ball. Since this is the intersection of a convex polytope with a ball, the maximum inscribed ball with center $x_c$ and radius $r$ is given by the optimal solution of the following conic program,
\begin{equation}
    \max\ r,\quad \text{subject to}: a_i^Tx_c + r||a_i|| \leq b_i,\, ||x_c|| \leq 1-r ,
\end{equation}
where $\Delta$ is defined as the intersection of the half-spaces $a_i^Tx \leq b_i,\ i=1,\dots, d+1$. Then, any vertex $v$ of $\Delta$ that also belongs to the connected component of graph $G$ (defined in Section~\ref{sec:geometric_modeling}) 
 corresponds to a component $\cK_i$
 and thus we can use it to obtain a point in $\cK_i$.
 	For this, we consider the segment defined by $v$ and $x_c$, then, the intersection of the segment with $\Sd$ lies in $\cK_i$.
\vspace{0.2cm}\\
\noindent
\textbf{Convergence to the target distribution.} We assess the quality of our results by employing a widely used Markov Chain Monte Carlo diagnostic, namely the potential scale reduction factor (PSRF)~\cite{Gelman92}. In particular, we compute the PSRF for each univariate marginal of the sample that
both GCW and ReGCW output. Following \cite{Gelman92}, and plenty of other works, a convergence is satisfying according
to PSRF when all the marginals have PSRF smaller than $1.1$.%\todo{add link to appendix "Correctness of sampling methods on real data"}
\vspace{0.2cm}\\
\noindent
\textbf{mVF restricted to a segment (GCW).} To sample from mVF using GCW, in each step of the random walk, we have to sample from the mVF constrained on a part of a great cycle. Thus, the goal is to sample from the univariate distribution $\pi_{\ell}(\theta) \propto e^{\alpha (\mu^Tp\cos\theta + \mu^Tv\sin\theta)}$, 
%\todo{We need to define all the parameters.}
for $\theta\in[\theta_1, \theta_2]$. To perform this operation, we use Metropolis-Hastings~\cite{Chib95}. The proposal probability density we use is the uniform distribution on a segment of length $(\theta_2 - \theta_1) / 3$ with the median being the current Markov point of GCW. We set the walk length (the number of Markov points to burn until storing a point) equal to a fixed value, i.e., $10$. We found that this value is an efficient choice; in our experiments the Markov point of GCW changes after the $95\%$ of the total number of steps we performed in our experiments. Thus, the empirical probability of stacking at a certain point is quite small.
\vspace{0.2cm}\\
\noindent
\textbf{Parameters of ReGCW.} To employ ReGCW (see section~\ref{sec:regcw}), we have to efficiently select values for the
parameter $\tau$ that controls the length of the trajectory in each step, for
the maximum number of reflections per step $\rho$, and for the walk length
of the random walk. We have experimentally found that setting the walk length equal to $1$, is the fastest choice so that the empirical distribution converges to the uniform distribution. To set $\tau$ for a component $\cK_i$ we sample $20d$ points with GCW with uniform target distribution. Then, we set $\tau$ equal to the length of the maximum geodesic chord in Equation~\ref{eq:bo_cgw_2} computing in those $20d$ steps of GCW. 
For the maximum number of reflections, we experimentally found that
$\rho = 100d$ is violated in less than $0.1\%$ of the total number of ReGCW steps in our experiments.
%\vspace{0.2cm}\\
%\noindent
%\textbf{Parameters for practical volume approximation.} 

\subsection{Construction of volatility-constrained random portfolios and backtesting framework}
\label{subsec:backtest_framework}

We apply ReGCW sampling to construct sets of portfolios having a predefined variance and investigate whether the out-of-sample performance of so constructed portfolios varies in a systematic way with the variance level (do portfolios with higher variance deliver higher, lower, or equal returns?). Out-of-sample means that we analyze the future (ex-post) performance of portfolios formed with volatility targets derived from past (ex-ante) stock price information and that, at every point of the back-testing procedure, we only use information that was effectively available at that point in time. Starting in March 2002, the earliest possible date for our data set described below, the implementation consists of a three-step process that is applied initially and repeated every three months in order to account for new information about stock risks and index composition. The three steps are data cleaning, covariance estimation, and sampling. 
In total, the quarterly reviews amount to $80$ time points where portfolios are rebuilt by first, cutting out a historical data sample of five years of cleaned weekly returns to estimate the covariance matrix which defines the ellipsoidal portfolio-variance level sets. Our choice of estimator is the non-linear shrinkage estimator \cite{bib:LedoitWolf2020} which has been shown to possess desirable properties in large-dimensional setups and is guaranteed to produce non-singular matrices and thus non-degenerate ellipsoids. Given the covariance matrix, five variance targets are computed from volatility-sorted quintile portfolios with equal weighting of within-quintile assets. From each of the five volatility level sets, $1000$ portfolios are sampled.\footnote{Increasing the number of samples is not a bottleneck computationally (the complexity is linear). However, we found no additional economic value by adding more samples since qualitatively, our results did not change.} The investments are held over the following three months until the process is repeated. Ultimately, we arrive at a total of $5000$ backtested portfolio price paths capturing the profits and losses endured over a period spanning from March 2002 to Dezember 2021 by randomly concatenating time series within each volatility cluster at the 80 rebalancing dates.

The entire procedure is repeated several times in order to control for size and sector effects by distinguishing between the 50\% smallest and largest companies and by further labeling companies as either defensive or cyclical according to the sector classification methodology employed by MSCI\footnote{Defensive sectors: Staples, Utilities, Energy and Health Care. Cyclical sectors: Financials, Real Estate, Information Technology, Discretionary, Industrials, Materials, Communication. For more detailed information see: 
\url{https://www.msci.com/eqb/methodology/meth_docs/MSCI_Cyclical_and_Defensive_Sectors_Indexes_Methodology_Nov18.pdf}} building on the Global Industry Classification Standards (GICS)\footnote{see \url{https://www.msci.com/our-solutions/indexes/gics}}.

\subsubsection{Data}\label{sec:data}

The data basis for our empirical study consists of two large universes of stock price series of companies covered by the MSCI USA and the MSCI Europe indices\footnote{see https://www.msci.com/our-solutions/indexes/developed-markets} which encompass large and mid-cap equities traded in the U.S. and across 15 developed countries in Europe. 
%As of November 2021, the list of index constituents counts 390 companies.

Our estimations are based on discrete weekly total\footnote{Returns, i.e., the percentage changes in prices from time $t-1$ to $t$, are termed \textit{total} when adjusted for dividends (i.e., dividends are re-invested).} returns using Wednesday closing prices denoted in local currencies. We use local currencies since we do not want to model any foreign exchange rates which would add another layer of volatility to the return series when expressed in a particular base currency (except for stocks that are already denominated in that currency). The currency effect is only relevant for the European market.
%Using weekly returns (rather than daily) avoids problems with non-synchronous trading of companies in Israel at weekends and arguably removes some noise.
Out-of-sample simulations are based on discrete daily total returns denoted in U.S. Dollars. Finally, we calculate all descriptive statistics and significance tests on discrete monthly returns as is customary in the financial industry.

The data cleaning process starts with adjustments for past corporate actions such as dividends, mergers and acquisitions, name changes, and other corporate actions. In addition, stocks that do not have enough history are excluded from the sample. To be included in the study, stocks need a consistent price history of five years, the equivalent of 260 weekly returns, without any gaps larger than two weeks. Further, otherwise, illiquid stocks are removed from the investable universe. As a threshold, we require a median trading volume over the previous 365 days to be above USD 1.5 million. We do this because, on the one hand, such illiquid stocks are not easily tradable and therefore would lead to a large implementation shortfall (i.e., the difference between a simulated performance and one obtained from real investments) and on the other hand, such companies display artificially low volatility due to a lack of trading and not because they are not risky.
The cleaning process is necessary to ensure that at every point in time, the investable universe only contains information that was effectively available at that point in time and to avoid any positive survivorship bias. Reference index membership over the full sample period is therefor not a requirement.

%\subsubsection{Correctness of sampling methods on real data}

%We provide experimental evidence that the samples generated by our MCMC methods in section~\ref{sec:algorithms_modeling}. More precisely, select certain time periods that the performance of the market is known \dots \todo{Should we add something here?}

%\subsubsection{Implementation - fees}

%When simulating the performance of the sampled strategies over the entire back-testing period (i.e., over all quarterly sub-periods) we apply a variable fee of 0.45\% on the portfolio's weight changes (turnover) occuring at the rebalancing dates. In addition, all simulated strategies pay a fixed management fee of 1\% per annum. 

\section{Results}\label{sec:results}

Our experimental investigation of the low-volatility effect over the past nearly twenty years using introduced geometric tools allows for a series of interesting conclusions. In particular, the sampling-based method provides insights into the distribution of risk and return statistics as visualized in Figure ~\ref{fig:risk_return_usa_europe}. The charts show the risk-return profiles of the five clusters of backtested sampling-based portfolios for the U.S. (left plot) and European (right plot) markets where clusters are color-coded by increasing variance from green to red. Each point indicates the annualized performance statistics of a backtested strategy using the procedure described in section~\ref{subsec:backtest_framework}. The points per volatility group are overlaid by bivariate (nonparametric kernel) density contours lines. The black dots represent the cluster averages whereas blue dots depict the performance of the classical sorting-based quintile portfolios. The light blue dot further shows the performance of a minimum-variance portfolio backtest and the black square tells the performance of the capitalization-weighted market index.
Inspection of Figure ~\ref{fig:risk_return_usa_europe}, as well as Figure~\ref{fig:risk_return_usa_4plots} and~\ref{fig:risk_return_europe_ex_ch_4plots} which display the same insights for the four sub-markets filtering for large defensive, small defensive, large cyclical, and small cyclical companies, reveals the following results.

\begin{figure}[t] 
\centering
\includegraphics[width=1\linewidth]{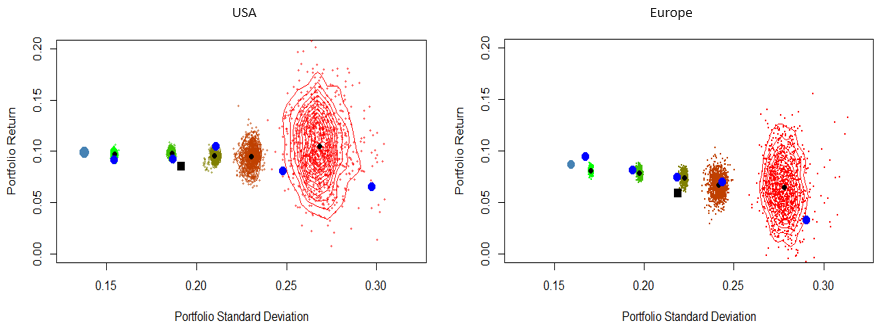}
\caption{Risk-return profiles of the five clusters of backtested sampling-based portfolios for the U.S. (left) and European (right) market where clusters are color-coded by increasing variance from green to red. The black dots represent the cluster averages whereas blue dots depict the performance of the classical sorting-based quintile portfolios. The light blue dot further shows the performance of a minimum-variance portfolio backtest and the black square tells the performance of the capitalization-weighted market index.}
\label{fig:risk_return_usa_europe}
\end{figure}

\begin{figure}[t] 
\centering
\includegraphics[width=1\linewidth]{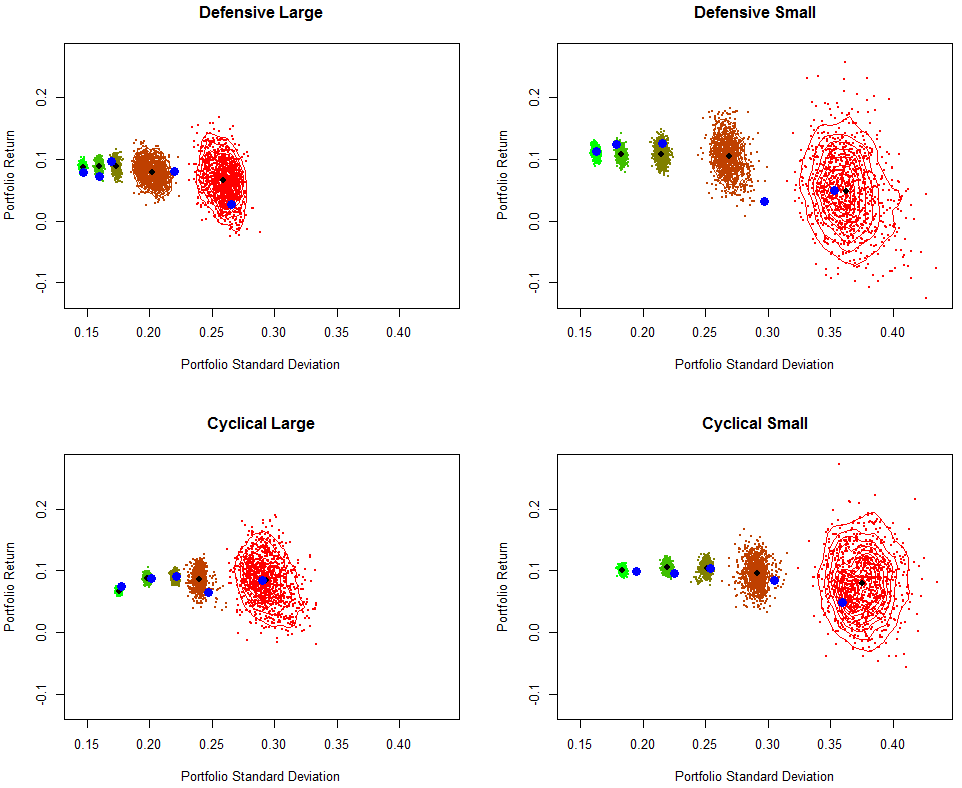}
\caption{Risk-return profiles of the five clusters of backtested sampling-based portfolios for the U.S. market are grouped into four sub-markets according to company size and sector belongingness.}
\label{fig:risk_return_usa_4plots}
\end{figure}

\begin{figure}[t] 
\centering
\includegraphics[width=1\linewidth]{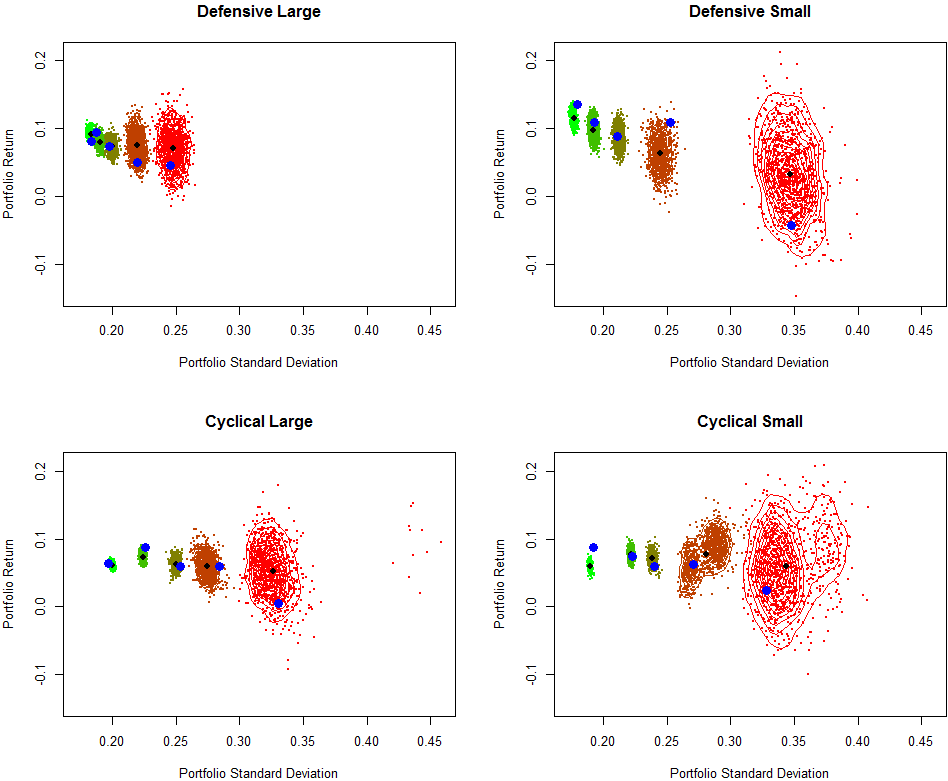}
\caption{Risk-return profiles of the five clusters of backtested sampling-based portfolios for the European market are grouped into four sub-markets according to company size and sector belongingness.}
\label{fig:risk_return_europe_ex_ch_4plots}
\end{figure}

\begin{figure}[t] 
\centering
\includegraphics[width=0.7\linewidth]{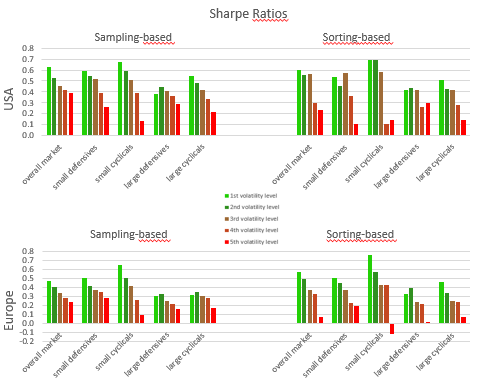}
\caption{Sharpe ratios for the average sampling-based backtests and sorting-based backtests (quintile portfolios) in the U.S. and Europe, for the overall markets and the four sub-universes.}
\label{fig:sharpe_ratios_usa_europe}
\end{figure}

%Figures ~\ref{fig:xxx} to ~\ref{fig:xxx} display the Sharpe ratio and average return of backtested sorting-based and average sampling-based strategies per (sub-) universe, again color-coded from green to red for the increasing volatility levels, together with the performance of the minimum variance portfolio in light blue and the capitalization weighted index in black. 

First, we can confirm the presence of the low-volatility effect in both universes and corresponding sub-markets. Low-volatility portfolios historically delivered higher risk-adjusted returns than high-volatility portfolios, on average. However, the anomalous pattern is much more pronounced in Europe where we find a monotone decrease in return with increasing variance, i.e., very much the opposite of what would be predicted by the CAPM. In the U.S. market, the negative relation only holds for risk-adjusted performance measures like the Sharpe ratio (the ratio of average return over volatility). The risk-return pattern is essentially flat with a slight increase in returns for the highest volatility level, an effect that could not be seen in the introductory example (see Figure~\ref{fig:risk_return_100y}) using the long-term U.S. dataset. In fact, the effect is to a large extent attributable to the stellar performance of mega-cap tech firms following the crash of March 2020 induced by the outbreak of the Covid-19 pandemic. If we omit the last two years in the analysis, i.e., stopping backtests at the end of 2019, the effect weakens and one finds a picture close to the one in the introductory example where the average of the highest volatility portfolios has the lowest cumulative returns while the largest cumulative returns are obtained by portfolios with moderately large variance (quintiles $2$ and $3$) and not by the low-volatility portfolio (quintile $1$). 

Detailed performance statistics are reported in Tables ~\ref{tab:annualized_returns_usa} to 
~\ref{tab:annualized_sharpe_ratios_europe}. They give the annualized returns, standard deviations, and Sharpe ratios per (sub-) universe for the five sorting-based quintile portfolio backtests and the average of the corresponding sampling-based portfolio backtests together with performance measures of the minimum variance portfolio and the capitalization-weighted index. Figure~\ref{fig:sharpe_ratios_usa_europe} further visualizes the decrease in Sharpe ratio with increasing variance, which is present in all (sub-) markets for both, sampling- and sorting-based backtests.

Second, Figures~\ref{fig:risk_return_usa_europe}, \ref{fig:risk_return_usa_4plots} and \ref{fig:risk_return_europe_ex_ch_4plots} show that the variation in portfolio returns increases with the increasing ex-ante variance target and is largest for the highest variance group which contains both, the worst and the best performing portfolios. This shows that the construction method, i.e., the choice of weighting to form the volatility-targeting portfolios can heavily impact the final outcome. Further, also the spread in realized volatility levels increases with higher volatility targets, meaning that there is a larger estimation error for such portfolios.

Third, we observe that inference from the classical sorting-based method can be rather misleading in that resulting risk-return statistics can deviate substantially from the average sampling-based portfolio characteristics. Moreover, cross-checking the different sub-universes the deviation seems to be non-systematic. In order to evaluate the likelihood of finding the sorting-based risk-return vector within the cloud of sampling-based performance vectors we parameterize the risk-return distribution by fitting a log-concave model. We use the non-parametric model in~\cite{Cule10} and its implementation of {\tt LogConcDEAD} package~\cite{Culejss09} that fits on the data a log-concave density function which logarithm is a tenant function and its support is the convex hull of the cluster. The fitted density gives us the mode of the risk-return distribution (i.e.\ the most likely risk-return vector). Then, we compute the measures of a small rectangle centered on i) the mode, ii) the empirical cluster average and iii) the risk-return vecotors of the sorting based quintile portfolios. In all $10$ (sub-) universes, we take the volume of the small rectangle to be 1$\%$ of the volume of the support of the distribution. The plot in Figures~\ref{fig:dist_europe} and~\ref{fig:dist_usa} illustrate the PDF computed by {\tt LogConcDEAD} package and the quintile portfolios for each volatility level that corresponds to the plots in Figure~\ref{fig:risk_return_usa_europe}. Moreover, Table~\ref{tab:quintile_mode_prob} reports the measures of the rectangles. We notice that the measure of the area around the average is almost equal to that of the area around the mode, which is strong empirical evidence that the model in~\cite{Cule10} is a reasonable choice to evaluate the likelihood of finding a the performance vector of a particular backtest within each volatility cluster. Next, there is only one case ---Europe/4th volatility level--- where the rectangular area around the sorting-based quintile portfolio statistics achieves almost the same probability as the one around the average or the mode. In $4$ cases the performance vector of the quintile portfolio is even outside the convex hull of the cluster, i.e., the measure is zero. In another $4$ cases, the probability of the area around the quintile portfolio statistics is more than $10$ times smaller than that of the average or the mode and in one case it is $4$ times smaller. Consequently, according to the log-concave model in~\cite{Cule10}, in most of the cases it is unlikely to find the performance results obtained from the sorting-based quintile portfolio backtests among the results of the sampling-based backtests within a volatility cluster. Moreover, there is not any pattern with respect to the volatility levels that would indicate when it is likely or not for the quintile portfolio performance to be found.

By design, the sorting-based portfolio is a particular solution to the sampling-based set. Hence, our finding that the sorting-based backtest statistics can vary strongly from the cluster averages may be surprising. However, the reason as to why the sorting-based quintile portfolio statistics are sometimes far outlying the distributions of sampling-based portfolio statistics has a geometric explanation. Sorting-based quintile portfolios with equal weighting form centroidal points on faces of the simplex, i.e., boundary points of both the simplex and the ellipsoid component of a volatility level (see section~\ref{sec:algorithms_modeling}). Typically, the areas around these points are highly unlikely to be sampled as they often lie in a low-volume area of the sampling space, e.g., a corner. Hence, quintile portfolios should not be used to represent the group of portfolios with a specific volatility level.

Looking at the backtests on the different (sub-) universes we find that, overall, the low-volatility effect in terms of return gap between high-volatility and low-volatility is overestimated by the sorting-based approach, and particularly so in the European markets. Moreover, the variation in return among the five volatility clusters is more erratic for the sorting-based method than for the mode of the sampling-based portfolio return distributions.
Not surprising is the observation that stocks from small and cyclical companies are more volatile than those of large and defensive companies (irrespective of whether we consider sorting- or sampling-based portfolios).

A fourth insight, which however is not apparent visually, is that the risk-return relation of portfolios within volatility cluster are mostly negative. I.e., even for portfolios that, by construction, have the same ex-ante volatility, one observes that ex-post, those with lower volatility generally lead to a higher risk-adjusted return. This shows in terms of negative coefficients of the correlation between annualized returns and standard deviations of backtests within a volatility cluster, i.e., a sort of within-cluster low-volatility anomaly.

In the next step, we analyze the statistical significance in the difference of Sharpe ratios between the lowest and highest volatility portfolios. We employ a Sharpe ratio test \cite{bib:LedoitWolf2008} which accounts for time series structures in the data by employing heteroscedasticity and autocorrelation consistent (HAC) estimates of standard error. In the U.S., the null of equal Sharpe ratios among the sorting-based portfolios can not be rejected at the $5$\% significance level. This also holds for any of the analyzed sub-markets except for the group of small defensives stocks. In Europe, the results are more mixed. The overall market as well as the sub-markets small defensives and small cyclicals show a significant Sharpe ratio difference, while large companies on the two-sector groups do not. 
Using the sampling-based simulations, we then run the test on all pairwise combinations of high minus low volatility simulations and count the number of significant t-statistics. We get less than $10$\% significant test results in the U.S. and about $15$\% in Europe, which, at first glance, does not exactly speak for the existence of an anomaly. However, if we discriminate the pairwise differences at a level of zero, i.e., we divide the sample between positive and negative Sharpe ratio differences, we observe that
$98$\% and $99$\% of the Sharpe ratio differences are positive in the U.S. and European markets respectively and all significant test results come from the subset with positive Sharpe ratio differences. I.e., we only find significant positive differences and no significant negative differences. Figure~\ref{fig:density_of_pvalues_europe_ex_ch_and_usa} further shows the distribution of p-values of the pairwise test statistics for series with positive Sharpe ratio difference in green and negative Sharpe ratio difference in red for the U.S. (left plot) and European (right plot) markets. Details of the test statistics for the two universes and sub-markets are contained in Tables ~\ref{tab:sharpe_ratio_tests_usa} and ~\ref{tab:sharpe_ratio_tests_europe}. 

The picture on the more granular analysis of the sub-markets discriminating among sectors and company size shows a rather consistent (relative) pattern among the two universes. In both cases, the highest proportion of positive Sharpe ratios is found among small defensives with a proportion of significant p-values of almost 60\% in the U.S. (while the p-value obtained on the sorting-based backtests is barely significant) and almost 70\% in Europe. Further, both universes display the smallest proportion of positive Sharpe ratios among the group of large cyclicals.

We conclude that, while test statistics based on the sorting-based quintile portfolios do not provide an unambiguous picture, it is overwhelmingly clear from the descriptive analysis of sampling-based portfolios that low-volatility portfolios have delivered higher Sharpe ratios than high-volatility portfolios, irrespective of the weighting scheme used to form the portfolios.\\

\begin{figure}[t] 
\centering
\includegraphics[width=0.8\linewidth]{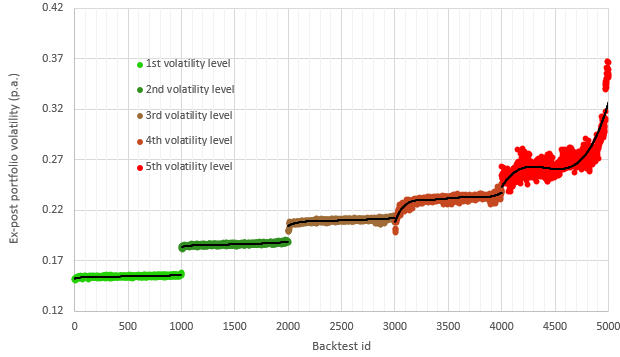}
\caption{Ex-post volatility of sampling-based backtests on the U.S. market where concatenation of random portfolios at the rebalancing dates is made according to a momentum sort. Points with the same color have the same in-sample volatility. Points per volatility cluster are overlayed by a polynomial regression curve. The curves are increasingly upward sloping with increasing volatility level meaning that backtests combined of the best past performers display (much) lower volatility than backtests composed of past losers.}
\label{fig:momentum_sorted_usa}
\end{figure}

As a final step, we repeat our analysis in a slightly modified version. Instead of randomly concatenating the sampling-based portfolio simulations at the rebalancing dates, we now order sampled portfolios within a volatility level by their last period return such that the currently best performing allocation matches with the previously best-performing portfolio, the currently second best with the previous second best, and so forth up to the current worst with the previous worst. The idea is that a concatenation based on a systematic ordering of the simulations according to a cross-sectional (in-sample) momentum criteria might produce a distinctive pattern out-of-sample. And indeed it does. Doing the concatenation over time as described, each of the $1000$ backtests per volatility level describes a dynamic strategy in the sense that the allocation changes every three months and the backtests are ranked by their quarterly performances. The first backtest concatenates all portfolios which ranked first in terms of quarterly returns. The pattern that emerges is a positive relation between in-sample rank and out-of-sample volatility. The backtest composed of the best performing portfolios has one of the lowest volatility over the entire nearly twenty years period while the backtest composed of the worst quarterly performers shows one of the highest volatility (always within a volatility-cluster). The pattern is very strong and shows up in all (sub-) markets. 
Portfolios with the same ex-ante volatility, when sorted by ex-ante momentum and concatenated over time, display a strong negative relation between the sorting-rank and ex-post volatility. The effect clearly shows in the increasingly upward sloping scatterplots in Figure~\ref{fig:momentum_sorted_usa} for the five color-coded volatility clusters in the U.S. market. In each volatility cluster, the backtests are ordered by rank (from best to worst) from left to right. Points are overlayed with a polynomial regression line to highlight the upward trend. Results for Europe are equivalent qualitatively, as can be seen in Figure~\ref{fig:mom_sorted_volas_europe_subsets}.
 which shows the behavior for the four European sub-markets.

%If we replace out-of-sample risk on the $y-$axis with the out-of-sample return we get the chance to detect another anomaly, namely the momentum effect: The backtest of past winners performs better in terms of cumulative returns than the backtest of past losers. 

The altered perspective provides further evidence in favor of the low-volatility anomaly. However, the economic rational and the causal nature of the very pronounced return-risk (rather than risk-return) relation is not clear to us at this stage. We leave the investigation to future research.

\begin{figure}[t] 
\centering
\includegraphics[width=1\linewidth]{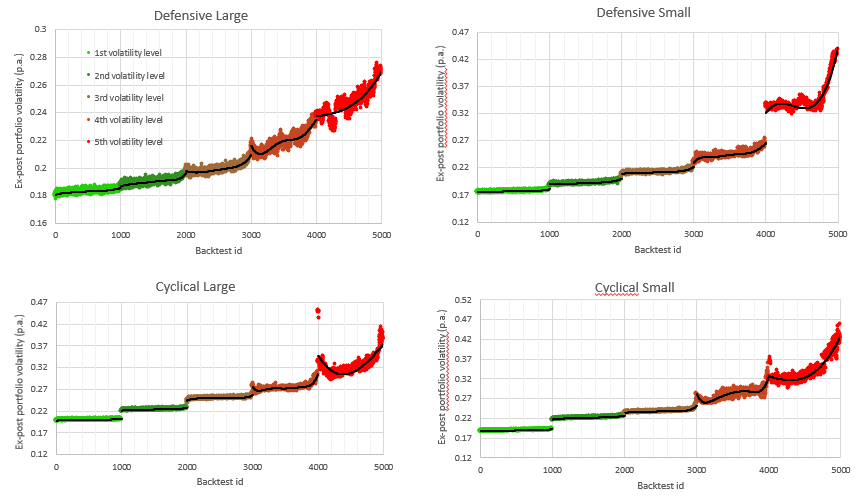}
\caption{Ex-post volatility of sampling-based backtests on the four subsets of the European market where concatenation of random portfolios at the rebalancing dates is made according to a momentum sort. Points with the same color have the same in-sample volatility. Points per volatility cluster are overlayed by a polynomial regression curve. The curves are increasingly upward sloping with increasing volatility level meaning that backtests combined of the best past performers display (much) lower volatility than backtests composed of past loser portfolios.}
\label{fig:mom_sorted_volas_europe_subsets}
\end{figure}

\begin{table}[t]
\centering
%\small
\footnotesize
\begin{tabular}{|c|c|ccccc|}
 \hline
 \multicolumn{2}{|c|}{USA} & \makecell{overall\\market} &
 \makecell{large\\defensives}&
 \makecell{small\\defensives} & \makecell{large\\cyclicals} &  \makecell{small\\cyclicals} \\
 \hline
 Sorting-based & p-values & 0.073 & 0.084 & 0.049 & 0.626 & 0.069 \\
 \hline
 \multirow{5}{*}{Sampling-based} & \makecell{Proportion of\\ significant p-values} & 0.114 & 0.254 & 0.567 & 0.004 & 0.119 \\
 \cdashline{2-7}
  & \makecell{Proportion of\\ positive SR} & 0.981 & 0.993 & 0.998 & 0.772 & 0.995 \\
  \cdashline{2-7}
  & \makecell{Proportion of significant\\ p-values on positive SR subset} & 0.116 & 0.256 & 0.569 & 0.006 & 0.120 \\
    \cdashline{2-7}
  & \makecell{Proportion of significant\\ p-values on negative SR subset} & 0 & 0 & 0 & 0 & 0 \\
 \hline
\end{tabular}
\caption{Sharpe ratio test results for the U.S. market.}
\label{tab:sharpe_ratio_tests_usa}
\end{table}

\begin{table}[!h]
\centering
\footnotesize
\begin{tabular}{|c|c|ccccc|}
 \hline
 \multicolumn{2}{|c|}{Europe} & \makecell{overall\\market} &
 \makecell{large\\defensives}&
 \makecell{small\\defensives} & \makecell{large\\cyclicals} &  \makecell{small\\cyclicals} \\
 \hline
 Sorting-based & p-values & 0.010 & 0.084 & 0.000 & 0.060 & 0.044 \\
 \hline
 \multirow{5}{*}{Sampling-based} & \makecell{Proportion of\\ significant p-values} & 0.152 & 0.140 & 0.678 & 0.019 & 0.028  \\
 \cdashline{2-7}
  & \makecell{Proportion of\\ positive SR} & 0.990 & 0.972 & 1 & 0.901 & 0.822 \\
  \cdashline{2-7}
  & \makecell{Proportion of significant\\ p-values on positive SR subset} & 0.153 & 0.144 & 0.678 & 0.021 & 0.033 \\
  \cdashline{2-7}
  & \makecell{Proportion of significant\\ p-values on negative SR subset} & 0 & 0 & 0 & 0 & 0.002 \\
 \hline
\end{tabular}
\caption{Sharpe ratio test results for the European market.}
\label{tab:sharpe_ratio_tests_europe}
\end{table}

\section{Discussion}\label{section:discussion}

We proposed a novel randomized geometric method that can be used to analyze stock markets for anomalies. The tools have allowed us to sample portfolios with a fixed variance and to investigate their performance characteristic over time, leading us to conclude that the low-volatility anomaly is very much present in the U.S. and Europe as well as in sub-markets discriminating for sector belongingness and company size. Further, we gained insights that the classical sorting-based approach forming quantile-portfolios badly represents the set of possible portfolios having a certain volatility and that therefore, one should be careful to base inference on the them.

The main limitation of our method is the "bottlenecks" that appear in the non-convex components (i.e.\ higher-dimensional analogs of the ones in Figure~\ref{fig:components_examples}) and do not allow us to perform computations with more complex data. 

Performance-wise our method can be enhanced with more efficient samplers such as Hamiltonian Monte-Carlo (HMC)~\cite{Byrne13} as well as by exploiting properties of the spherical patches. Also, general polytopes could be used instead of the simplex to increase the expressive power of our modeling. 

\begin{figure}[t] 
\centering
\includegraphics[width=0.9\linewidth]{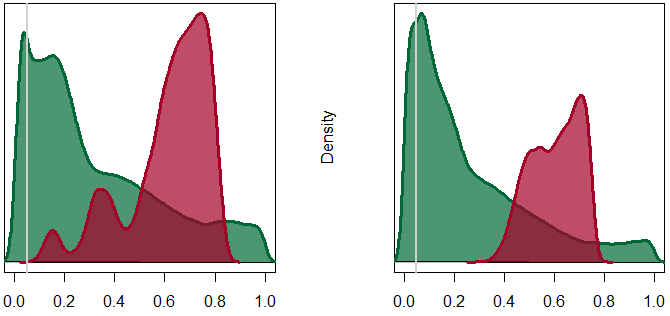}
\caption{Distributions of p-values of the Ledoit-Wolf Sharpe ratio test between low-volatility and high-volatility portfolios, grouped by positive (green) and negative (red) Sharpe ratio differences of low minus high volatility portfolio backtest in the U.S. (left) and European (right) market. The grey vertical lines mark the 5\% significance level.}
\label{fig:density_of_pvalues_europe_ex_ch_and_usa}
\end{figure}

\bibliography{references.bib}

\begin{thebibliography}{10}

\bibitem{Abbasi17}
Y.~Abbasi{-}Yadkori, P.L. Bartlett, V.~Gabillon, and A.~Malek.
\newblock Hit-and-run for sampling and planning in non-convex spaces.
\newblock In {\em Proc. 20th Intern. Conf. Artificial Intelligence \& Stat.\
  ({AISTATS)}}, pages 888--895, April 2017.

\bibitem{bib:AsnessEtAl2020}
C.~S. Asness, N.~J.~Gormsen A.~Frazzini, and L.~H. Pedersen.
\newblock Betting against correlation.
\newblock {\em Journal of Financial Economics}, 135:629 -- 652, 2020.

\bibitem{bib:AsnessFrazziniPedersen2019}
C.~S. Asness, A.~Frazzini, and L.~H. Pedersen.
\newblock Quality minus junk.
\newblock {\em Review of Accounting Studies}, 24:34 -- 112, 2019.

\bibitem{bib:BakerBradleyWurgler2011}
M.~Baker, B.~Bradley, and J.~Wurgler.
\newblock Benchmarks as limits to arbitrage: Understanding the low-volatility
  anomaly.
\newblock {\em Financial Analysts Journal}, 67:40 -- 54, 2011.

\bibitem{bib:BakerHaugen2012}
N.~L. Baker and R.~A. Haugen.
\newblock Low risk stocks outperform within all observable markets of the
  world.
\newblock {\em Working paper no. 2055431}, 2012.

\bibitem{bib:Banz1981}
R.~W. Banz.
\newblock The relationship between return and market value of common stocks.
\newblock {\em Journal of Financial Economics}, 9:3 -- 18, 1981.

\bibitem{bib:Black1993}
F.~Black.
\newblock Beta and return: Announcements of the ‘death of beta’ seem
  premature.
\newblock {\em The Journal of Portfolio Management}, 20:11 -- 18, 1973.

\bibitem{bib:BlackJensenScholes1972}
F.~Black, M.~C. Jensen, and M.~Scholes.
\newblock Asset pricing model: Some empirical tests.
\newblock pages 79 -- 121, 1972.

\bibitem{bib:BlitzvanVlietBaltussen2019}
D.~Blitz, G.~Baltussen, and P.~van Vliet.
\newblock When equity factors drop their shorts.
\newblock {\em Working paper, 2019}, 2019.

\bibitem{bib:BlitzPangvanVliet2013}
D.~Blitz, J.~Pang, and P.~van Vliet.
\newblock The volatility effect in emerging markets.
\newblock {\em Emerging Markets Review}, 16:31 -- 45, 2013.

\bibitem{bib:BlitzvanVliet2007}
D.~Blitz and P.~van Vliet.
\newblock The volatility effect.
\newblock {\em The Journal of Portfolio Management}, 34:102 -- 113, 2007.

\bibitem{Byrne13}
Simon Byrne and Mark Girolami.
\newblock Geodesic monte carlo on embedded manifolds.
\newblock {\em Scandinavian Journal of Statistics}, 40(4):825--845, 2013.

\bibitem{Smith93}
Claude J.~P. Bélisle, H.~Edwin Romeijn, and Robert~L. Smith.
\newblock Hit-and-run algorithms for generating multivariate distributions.
\newblock {\em Mathematics of Operations Research}, 18(2):255--266, 1993.

\bibitem{Cales18}
Ludovic Cal{\`e}s, Apostolos Chalkis, Ioannis~Z. Emiris, and Vissarion
  Fisikopoulos.
\newblock {Practical Volume Computation of Structured Convex Bodies, and an
  Application to Modeling Portfolio Dependencies and Financial Crises}.
\newblock In Bettina Speckmann and Csaba~D. T{\'o}th, editors, {\em 34th
  International Symposium on Computational Geometry (SoCG 2018)}, volume~99 of
  {\em Leibniz International Proceedings in Informatics (LIPIcs)}, pages
  19:1--19:15, Dagstuhl, Germany, 2018. Schloss Dagstuhl--Leibniz-Zentrum fuer
  Informatik.

\bibitem{bib:Carhart1997}
M.~M. Carhart.
\newblock On persistence in mutual fund performance.
\newblock {\em The Journal of Finance}, 52:57 -- 82, 1997.

\bibitem{CoolBod}
Apostolos Chalkis, Ioannis~Z. Emiris, and Vissarion Fisikopoulos.
\newblock Practical volume estimation by a new annealing schedule for cooling
  convex bodies, 2019.

\bibitem{chalkis2020volesti}
Apostolos Chalkis and Vissarion Fisikopoulos.
\newblock {volesti: Volume Approximation and Sampling for Convex Polytopes in
  R}.
\newblock {\em {The R Journal}}, 2021.

\bibitem{chalkis21}
Apostolos Chalkis, Vissarion Fisikopoulos, Elias Tsigaridas, and Haris
  Zafeiropoulos.
\newblock {Geometric Algorithms for Sampling the Flux Space of Metabolic
  Networks}.
\newblock In Kevin Buchin and \'{E}ric Colin~de Verdi\`{e}re, editors, {\em
  37th International Symposium on Computational Geometry (SoCG 2021)}, volume
  189 of {\em Leibniz International Proceedings in Informatics (LIPIcs)}, pages
  21:1--21:16, Dagstuhl, Germany, 2021. Schloss Dagstuhl -- Leibniz-Zentrum
  f{\"u}r Informatik.

\bibitem{Karthekeyan10}
Karthekeyan Chandrasekaran, Daniel Dadush, and Santosh Vempala.
\newblock Thin partitions: Isoperimetric inequalities and a sampling algorithm
  for star shaped bodies.
\newblock pages 1630--1645, 01 2010.

\bibitem{Chen18}
Y.~Chen, R.~Dwivedi, M.J. Wainwright, and B.~Yu.
\newblock Fast {MCMC} sampling algorithms on polytopes.
\newblock {\em Journal of Machine Learning Research}, 19(55):1--86, 2018.

\bibitem{Chib95}
Siddhartha Chib and Edward Greenberg.
\newblock Understanding the metropolis-hastings algorithm.
\newblock {\em The American Statistician}, 49(4):327--335, 1995.

\bibitem{bib:ClarkedeSilvaThorley2006}
R.~H. Clarke, H.~de~Silva, and S.~Thorley.
\newblock Minimumvariance portfolios in the us equity market.
\newblock {\em The Journal of Portfolio Management}, 33:10 -- 24, 2006.

\bibitem{bib:ClarkedeSilvaThorley2011}
R.~H. Clarke, H.~de~Silva, and S.~Thorley.
\newblock Minimum-variance portfolio composition.
\newblock {\em The Journal of Portfolio Management}, 37:31 -- 45, 2011.

\bibitem{Cong17}
Yulai Cong, Bo~Chen, and Mingyuan Zhou.
\newblock {Fast Simulation of Hyperplane-Truncated Multivariate Normal
  Distributions}.
\newblock {\em Bayesian Analysis}, 12(4):1017 -- 1037, 2017.

\bibitem{Cousins14}
B.~Cousins and S.~Vempala.
\newblock Bypassing {KLS: Gaussian} cooling and an ${O}^*(n^3)$ volume
  algorithm.
\newblock In {\em Proc. ACM STOC}, pages 539--548, 2015.

\bibitem{Cousins15}
B.~Cousins and S.~Vempala.
\newblock A practical volume algorithm.
\newblock {\em Mathematical Programming Computation}, 8(2), Jun 2016.

\bibitem{Culejss09}
Madeleine Cule, Robert~B. Gramacy, and Richard Samworth.
\newblock Logconcdead: An r package for maximum likelihood estimation of a
  multivariate log-concave density.
\newblock {\em Journal of Statistical Software}, 29(2):1–20, 2009.

\bibitem{Cule10}
Madeleine Cule, Richard Samworth, and Michael Stewart.
\newblock Maximum likelihood estimation of a multi-dimensional log-concave
  density.
\newblock {\em Journal of the Royal Statistical Society: Series B (Statistical
  Methodology)}, 72(5):545--607, 2010.

\bibitem{Davidson18}
Tim~R. Davidson, Luca Falorsi, Nicola~De Cao, Thomas Kipf, and Jakub~M.
  Tomczak.
\newblock Hyperspherical variational auto-encoders.
\newblock In {\em UAI}, 2018.

\bibitem{Diaconis13}
Persi {Diaconis}, Susan {Holmes}, and Mehrdad {Shahshahani}.
\newblock {Sampling from a manifold}.
\newblock In {\em Advances in modern statistical theory and applications. A
  Festschrift in honor of Morris L. Eaton}, pages 102--125. Beachwood, OH: IMS,
  Institute of Mathematical Statistics, 2013.

\bibitem{dieker2014stochastic}
Antonius~B Dieker and Santosh~S Vempala.
\newblock Stochastic billiards for sampling from the boundary of a convex set.
\newblock {\em Mathematics of Operations Research}, 40(4):888--901, 2015.

\bibitem{Dyer88}
M.~Dyer and A.~Frieze.
\newblock On the complexity of computing the volume of a polyhedron.
\newblock {\em SIAM Journal on Computing}, 17(5):967--974, 1988.

\bibitem{Emiris14}
{I.Z.} Emiris and V.~Fisikopoulos.
\newblock Practical polytope volume approximation.
\newblock {\em ACM Transactions of Mathematical Software, 2018},
  44(4):38:1--38:21, 2014.

\bibitem{bib:Falkenstein1994}
E.~G. Falkenstein.
\newblock Mutual funds. idiosyncratic variance, and asset returns.
\newblock {\em PhD thesis, Northwestern University}, 1994.

\bibitem{bib:FamaFrench1992}
E.~F. Fama and K.~R. French.
\newblock The cross-section of expected stock returns.
\newblock {\em The Journal of Finance}, 47:427 -- 465, 1992.

\bibitem{bib:FamaFrench2015}
E.~F. Fama and K.~R. French.
\newblock A five-factor asset pricing model.
\newblock {\em Journal of Financial Economics}, 116:1 -- 22, 2015.

\bibitem{bib:FamaMacBeth1973}
E.~F. Fama and J.~D. MacBeth.
\newblock Risk, return, and equilibrium: Empirical tests.
\newblock {\em Journal of Political Economy}, 81:607 -- 636, 1973.

\bibitem{bib:FrazziniPedersen2014}
A.~Frazzini and L.~H. Pedersen.
\newblock Betting against beta.
\newblock {\em Journal of Financial Economics}, 111:1 -- 25, 2014.

\bibitem{Gelman92}
Andrew Gelman and Donald~B. Rubin.
\newblock Inference from {Iterative} {Simulation} {Using} {Multiple}
  {Sequences}.
\newblock {\em Statistical Science}, 7(4):457--472, 1992.
\newblock Publisher: Institute of Mathematical Statistics.

\bibitem{Genz09}
A.~Genz and F.~Bretz.
\newblock {\em Computation of Multivariate Normal and t Probabilities}.
\newblock Springer Publishing Company, Incorporated, 1st edition, 2009.

\bibitem{grattarola19}
Daniele Grattarola, Lorenzo Livi, and Cesare Alippi.
\newblock Adversarial autoencoders with constant-curvature latent manifolds.
\newblock {\em Applied Soft Computing}, 81:105511, 2019.

\bibitem{Gryazina14}
Elena Gryazina and Boris Polyak.
\newblock Random sampling: Billiard walk algorithm.
\newblock {\em European Journal of Operational Research}, 238(2):497 -- 504,
  2014.

\bibitem{eigenweb}
Ga\"{e}l Guennebaud, Beno\^{i}t Jacob, et~al.
\newblock {\em {Eigen} v3}, 2010.

\bibitem{HHPS2002}
W.~Hallerbach, C.~Hundack, I.~Pouchkarev, and J.~Spronk.
\newblock A broadband vision of the development of the dax over time.
\newblock Technical Report ERS-2002-87-F\&A, Erasmus University Rotterdam,
  2002.

\bibitem{bib:HaugenBaker1991}
R.~A. Haugen and N.~L. Baker.
\newblock The efficient market inefficiency of capitalization-weighted stock
  portfolios.
\newblock {\em The Journal of Portfolio Management}, 17:35 -- 40, 1991.

\bibitem{bib:HaugenHeins1975}
R.~A. Haugen and A.~J. Heins.
\newblock Risk and the rate of return on financial assets: Some old wine in new
  bottles.
\newblock {\em Journal of Financial and Quantitative Analysis}, 10:775 -- 784,
  1975.

\bibitem{Iyengar88}
S.~Iyengar.
\newblock Evaluation of normal probabilities of symmetric regions.
\newblock {\em SIAM Journal on Scientific and Statistical Computing},
  9(3):418--423, 1988.

\bibitem{bib:JegadeeshTitman1993}
N.~Jegadeesh and S.~Titman.
\newblock Returns to buying winners and selling losers: Implications for stock
  market efficiency.
\newblock {\em The Journal of Finance}, 48:65 -- 91, 1993.

\bibitem{Jeter05}
Sheldon Jeter.
\newblock A handy tool for convenient error propagation analysis: A user form
  for error influence coefficients.
\newblock In {\em 2005 Annual Conference}, number 10.18260/1-2--14693,
  Portland, Oregon, June 2005. ASEE Conferences.
\newblock https://peer.asee.org/14693.

\bibitem{Karney12}
Charles F.~F. Karney.
\newblock Algorithms for geodesics.
\newblock {\em Journal of Geodesy}, 87(1):43–55, Jun 2012.

\bibitem{Kurz15}
Gerhard Kurz and Uwe~D. Hanebeck.
\newblock Stochastic sampling of the hyperspherical von mises–fisher
  distribution without rejection methods.
\newblock In {\em 2015 Sensor Data Fusion: Trends, Solutions, Applications
  (SDF)}, pages 1--6, 2015.

\bibitem{bib:LedoitWolf2008}
O.~Ledoit and M.~Wolf.
\newblock Robust performance hypothesis testing with the sharpe ratio.
\newblock {\em Journal of Empirical Finance}, 15:850 -- 859, 2008.

\bibitem{bib:LedoitWolf2020}
O.~Ledoit and M.~Wolf.
\newblock Analytical nonlinear shrinkage of large-dimensional covariance
  matrices.
\newblock {\em The Annals of Statistics}, 48:3043 -- 3065, 2020.

\bibitem{Lee18}
Y.T. Lee and S.~Vempala.
\newblock Convergence rate of {R}iemannian {H}amiltonian {Monte Carlo} and
  faster polytope volume computation.
\newblock In {\em Proceedings of the 50th Annual ACM SIGACT Symposium on Theory
  of Computing}, STOC 2018, pages 1115--1121, 2018.

\bibitem{bib:Lintner1965}
J.~Lintner.
\newblock Security prices, risk and maximal gains from diversification.
\newblock {\em The Journal of Finance}, 20:587 -- 615, 1965.

\bibitem{bib:LiuStambaughYuan2018}
J.~Liu, R.~F. Stambaugh, and Y.~Yuan.
\newblock Absolving beta of volatility’s effects.
\newblock {\em Journal of Financial Economics}, 128:1 -- 15, 2018.

\bibitem{Lovasz06}
L.~Lov\'{a}sz and S.~Vempala.
\newblock Hit-and-run from a corner.
\newblock {\em SIAM Journal on Computing}, 35(4):985--1005, 2006.

\bibitem{Lovasz06b}
L.~Lovász and S.~Vempala.
\newblock Simulated annealing in convex bodies and an {O}$^*(n^4)$ volume
  algorithms.
\newblock {\em J. Computer \& System Sciences}, 72:392--417, 2006.

\bibitem{Mangoubi19}
O.~{Mangoubi} and N.~K. {Vishnoi}.
\newblock Faster polytope rounding, sampling, and volume computation via a
  sub-linear ball walk.
\newblock In {\em 2019 IEEE 60th Annual Symposium on Foundations of Computer
  Science (FOCS)}, pages 1338--1357, 2019.

\bibitem{Mardia75}
K.~V. Mardia.
\newblock Distribution theory for the von mises-fisher distribution and its
  application.
\newblock In G.~P. Patil, S.~Kotz, and J.~K. Ord, editors, {\em A Modern Course
  on Statistical Distributions in Scientific Work}, pages 113--130, Dordrecht,
  1975. Springer Netherlands.

\bibitem{bib:Markowitz1952}
H.~Markowitz.
\newblock Portfolio selection.
\newblock {\em The Journal of Finance}, 7:77 -- 91, 1992.

\bibitem{boostrandom}
Jens Maurer and Steven Watanabe.
\newblock Boost random number library.
\newblock Software, 2017.

\bibitem{bib:MillerScholes1972}
M.~H. Miller and M.~Scholes.
\newblock Rates of return in relation to risk: A reexamination of some recent
  findings.
\newblock pages 47 -- 78, 1972.

\bibitem{bib:Mossin1966}
J.~Mossin.
\newblock Equilibrium in a capital asset market.
\newblock {\em Econometrica}, 34:768 -- 783, 1966.

\bibitem{Narayanan06}
Hariharan Narayanan and Partha Niyogi.
\newblock Sampling hypersurfaces through diffusion.
\newblock In {\em Neural Information Processing Systems (NIPS}, page~7, 2006.

\bibitem{ortiz2021}
Joaquim Ortiz-Haro, Jung-Su Ha, Danny Driess, and Marc Toussaint.
\newblock Structured deep generative models for sampling on constraint
  manifolds in sequential manipulation.
\newblock In {\em 5th Annual Conference on Robot Learning (CoRL)}, 2021.

\bibitem{bib:PlyakhaUppalVilkov2014}
Y.~Plyakha, R.~Uppal, and G.~Vilkov.
\newblock Equal or value weighting? implications for asset-pricing tests.
\newblock {\em Working paper}, 2014.

\bibitem{PST04}
I.~Pouchkarev, J.~Spronk, and J.~Trinidad.
\newblock Dynamics of the spanish stock market through a broadband view of the
  {IBEX} 35 index.
\newblock {\em Estudios Econom. Aplicada}, 22(1):7--21, 2004.

\bibitem{reisinger10}
Joseph Reisinger, Austin Waters, Bryan Silverthorn, and Raymond~J. Mooney.
\newblock Spherical topic models.
\newblock In {\em Proceedings of the 27th International Conference on
  International Conference on Machine Learning}, ICML'10, page 903–910,
  Madison, WI, USA, 2010. Omnipress.

\bibitem{bib:RosenbergReidLanstein1985}
B.~Rosenberg, K.~Reid, and R.~Lanstein.
\newblock ersuasive evidence of market inefficiency.
\newblock {\em The Journal of Portfolio Management}, 11:9 -- 16, 1985.

\bibitem{Schellenberger09}
J.~Schellenberger and B.O. Palsson.
\newblock Use of randomized sampling for analysis of metabolic networks.
\newblock {\em The Journal of biological Chemistry}, 284 9:5457--61, 2009.

\bibitem{bib:Sharpe1964}
W.~F. Sharpe.
\newblock Capital asset prices: A theory of market equilibrium under conditions
  of risk.
\newblock {\em The Journal of Finance}, 19:425 -- 444, 1964.

\bibitem{Smith84}
Robert~L. Smith.
\newblock Efficient monte carlo procedures for generating points uniformly
  distributed over bounded regions.
\newblock {\em Operations Research}, 32(6):1296--1308, 1984.

\bibitem{Somerville98}
P.N. Somerville.
\newblock Numerical computation of multivariate normal and multivariate-t
  probabilities over convex regions.
\newblock {\em Journal of Computational and Graphical Statistics},
  7(4):529--544, 1998.

\bibitem{tacchi2020stokes}
Matteo Tacchi, Jean Lasserre, and Didier Henrion.
\newblock Stokes, gibbs and volume computation of semi-algebraic sets.
\newblock {\em arXiv preprint arXiv:2009.12139}, 2020.

\bibitem{tacchi2022exploiting}
Matteo Tacchi, Tillmann Weisser, Jean~Bernard Lasserre, and Didier Henrion.
\newblock Exploiting sparsity for semi-algebraic set volume computation.
\newblock {\em Foundations of Computational Mathematics}, 22(1):161--209, 2022.

\bibitem{bib:vanVlietdeKoning2017}
P.~van Vliet and J.~de~Koning.
\newblock High returns from low risk: A remarkable stock market paradox.
\newblock 2017.

\bibitem{vempala_survey_2005}
Santosh Vempala.
\newblock Geometric random walks: a survey.
\newblock {\em Combinatorial and Computational Geometry}, pages 573--612, 2005.

\bibitem{venzke2019}
A.~Venzke, D.K. Molzahn, and S.~Chatzivasileiadis.
\newblock Efficient creation of datasets for data-driven power system
  applications.
\newblock {\em Electric Power Systems Research}, 190:106614, 2021.

\bibitem{bib:Walkshausl2014}
C.~Walkshäusl.
\newblock International low-risk investing.
\newblock {\em The Journal of Portfolio Management}, 41:45 -- 56, 2014.

\end{thebibliography}

\newpage
\appendix
\section{Appendix}
\subsection{Membership oracle for a connected component}\label{appnd:mem_oracle}

\begin{lem}
Given a body $\cK = \Sd \cap \Delta$ and a point $p\in\Sd$ to decide if $p\in\cK$ costs $O(d^2)$ operations.
\end{lem}
\begin{proof}
The Figure~\ref{fig:membership} describes the steps of the algorithm in 2D. This procedure can be generalized for any dimension $d$. In particular, we consider the half-line $l_{op}$ which starts from the origin and passes through $p$. Let $l_{op}$ to intersect $\partial \Delta$ at point $q$. Let also $u$ one vertex of $\Delta$ that $q$ has visual contact with, i.e., the segment defined by $q$ and $u$ does not intersect $\Sd$. Then, due to convexity, $p$ belongs to the component of $\cK$ that corresponds to the component of graph $G$ that contains $u$.

The computation of $q$ takes $O(d^2)$ operations and the detection of the vertex $u$ takes $O(d^2)$ operations too.
\end{proof} 

\begin{figure}[h] \centering
\includegraphics[width=0.85\linewidth]{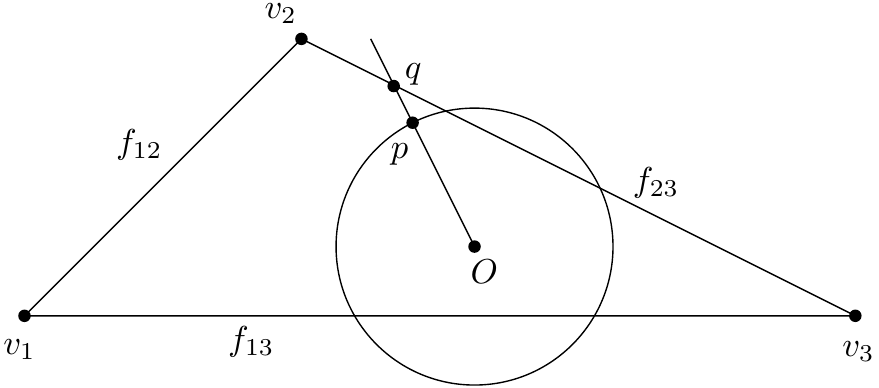}
\caption{Illustration of how the membership oracle works. In this example, there are two connected subsets of $\mathcal{S}_{1}$. We would like to answer in which component the point $p$ belongs. The half-line $l_{op}$ intersects $\partial P$ on $q$ which lies in the facet $f_{23}$. The point $q$ has visible contact with the vertex $v_2$, and thus, belongs to the connected component of $\mathcal{S}_1$ that corresponds to the connected component $\{v_1,v_2\}$ of graph $G$. \label{fig:membership}}
\end{figure}

\section{Sample uniformly with Great Cycle Walk}\label{appnd:gcw_uniform}

To sample uniformly a connected component of $\cK$, GCW samples uniformly a point from $\ell(\theta)$ in each step. 
In general, to sample uniformly from a parametric curve, we have to sample from the univariate probability density induced by the norm of the derivative of the curve. 
In our case, the curve is a (part of a) great circle and so GCW samples from
\begin{equation}
\phi(\theta) \propto \|\ell'(\theta)\|_2 = \| -p \sin \theta  + v \cos \theta \, \|_2 = 1 ,\ \theta \in [\theta^-,\theta^+],
\end{equation}
which is the uniform distribution over the segment $[\theta^-,\theta^+]$. 
Consequently, starting from $p$ the next Markov point is $p \cos \tilde{\theta}  + v \sin \tilde{\theta}$, for $\tilde{\theta}$ sampled uniformly from $[\theta^-,\theta^+]$.

\section{Proofs}

\subsection{Theorem~\ref{thm:GCW}} 
\begin{proof}
We build upon the results and the methodology in~\cite{Smith93}. 
Let $\mathcal{D}$ be a set of subsets of $K$ and $P(p,A)$, with $p\in K$ and $A\in\mathcal{D}$, denotes the one step transition probabilities (Markov kernel) of CGW algorithm. 
We will prove that $P$ is reversible with respect to $\pi$, that is
\begin{equation}\label{eq:reversibility}
    \int_A P(p,B)\pi(dp) = \int_B P(r,A)\pi(dr), \text{ for all }A,B\in\mathcal{D}
\end{equation}
Then, the stationarity of $\pi$ follows at once since we set $B = K$ in Equation~(\ref{eq:reversibility}) and we get,
\begin{equation}\label{eq:stationarity}
    \pi(A) = \int_{K} P(r,A) \pi(dr), \text{ for all } A\in\mathcal{D} ,
\end{equation}
which implies stationarity of $\pi$~\cite{vempala_survey_2005}. 

For $p,\ q\in K$ with $p^Tq = 0$ let,
\begin{equation}
    \Phi(p,q) = \{ \theta\in [-\pi, \pi]\ |\ p\cos\theta + q\sin\theta \in K \} .
\end{equation}
Let $f$ be the probability distribution function (PDF) of the constrained $\pi$ on a great cycle, then, let $f_{(p,q)}$ be the PDF on $[\theta^-, \theta^+]$ defined by
\begin{equation}
    f_{(p,q)}(\theta) = \left\{
\begin{array}{ll}
      \frac{f(p\cos\theta + q\sin\theta)}{\int_{\Phi(p,q)} f(p\cos t + q\sin t)\ dt} & \mbox{ if } \theta \in  \Phi(p,q),\\
      0, & \mbox{otherwise.}\\
\end{array} 
\right.
\end{equation}
Consider the random variables $Q$ and $U$,
 where $Q$ follows the uniform distribution over $\Sd \cap \mathcal{H}_p$, where $\mathcal{H}_p := \{ x\in\R^d\ |\ p^Tx = 0 \}$, $U$ follows the uniform distribution over $(0,1)$, and $F_{(p,q)}$ is the cumulative distribution function (CDF) of $f_{(p,q)}$. Then, the Markov kernel $P(p, A)$ is
\begin{equation}
    \begin{split}
         P(p, A) &= \Pr[F^{-1}_{(p, Q)}(U) \in A] \\
         & = \int_{\Sd \cap \mathcal{H}_p} \frac{1}{\vol(\Sd \cap \mathcal{H}_p)} \Pr[F^{-1}_{(p, q)}(U) \in A] dq \\
        & = \int_{\Sd \cap \mathcal{H}_p} \frac{1}{\vol(\Sd \cap \mathcal{H}_p)} \int_{\Phi(p,q)} \bm{1}_A(p\cos\theta + q\sin\theta) f_{(p,q)}(\theta) d\theta\ dq, 
    \end{split}
\end{equation}
where $\bm{1}_A(\cdot)$ is the indicator function of the set $A$ and $F^{-1}_{(p,q)}$ is the left-continuous inverse of CDF $F_{(p,q)}$.  
Let $W$ the conic bundle on $\mathcal{H}_p$ centered at $p$ with small spatial angle $d\phi$ that defines the geodesic trajectories starting from $p$ and have non-empty intersection with an infinitesimally small neighborhood $dr$ in $K$. Then, the Markov kernel becomes $P(p,A) = \int_A g(p,r) \pi(dr)$, where
\begin{equation}
    g(p,r) = \frac{\vol(W \cap \Sd \cap \mathcal{H}_p)}{\vol(\Sd \cap \mathcal{H}_p)\int_{\Phi_{pr}} p\cos t + q\sin t\ dt} ,
\end{equation}
where $\Phi_{pr} = \{ \theta\in [-\pi, \pi]\ |\ p\cos\theta + q\sin\theta \in \mathcal{C}(pr) \}$, for well-chosen $q\in\Sd \cap \mathcal{H}_p$ and $\mathcal{C}(pr)$ being the part of the great cycle defined by $p,\ r\in\Sd$ inside $K$. Notice that $g(p,r)$ is symmetric. 

Then, the left side in Equation~(\ref{eq:reversibility}) becomes $\int_A \int_B g(p,r)\pi(dr)\ \pi(dp)$ while the left side becomes $\int_B \int_A g(r,p) \pi(dp)\ \pi(dr)$. Reversibility then follows from Fubini's theorem and from the fact that $g(p,r)$ is symmetric.
\end{proof}

\subsection{Theorem~\ref{thm:gcw_convergence}} 
\begin{proof}
We build upon the results and the methodology in~\cite{Gryazina14}. 
Let $K_i$ a connected component of $K = \Delta \cap \Sd$.
The Theorem~2 in~\cite{Smith84} proves that if the transition density $r(q|p)$ exists and is symmetric as well as it is positive for all $p,q\in K_i$ then, the uniform distribution over $K_i$ is a unique stationary distribution, and it is achieved for any starting point $p\in K_i$.
To prove convergence to the uniform distribution we consider two cases: when $K_i$ is geodesically convex set and when $K_i$ is geodesically non-convex set.
In both cases, being at a Markov point $p^j$, the next Markov point $p^{j+1}$ is obtained with positive probability with less than $\rho + 1$ reflections.

For the first case, the existence of the probability density $r(q,p)$ for any $q,p\in K_i$ is implied when the transition probability from $p$ to an infinitesimally small neighborhood $dq$ is proportional to the volume of $dq$. %\footnote{Tolis: define infinitesimally small neighborhood $dq$}. 
Considering all possible piece-wise geodesic trajectories ---defined as pieces of great cycles--- that go from $p$ to $dq$, take those that perform $0 \leq k \leq \rho$ reflections. With this set of trajectories there is a conic bundle on the plane $H_p$ centered at $p$ with with small spatial angle $d\theta$ that define these trajectories. The area of reﬂection can be approximated as plain region, and thus, a reﬂection does not change the geometry of the bundle. Then,
\begin{equation}
    \Pr[\delta q|p] \propto \Pr[\delta\theta] \Pr[\delta L] , 
\end{equation}
where $\Pr[d\theta]$ is the probability of choosing the spatial angle (proportional to the volume of the base of the cone) and $\Pr[dL]$ is the probability
of choosing a certain trajectory length $L \in dL$. Thus, $\Pr[dq|p] \propto \vol(dq)$. For a geodesically convex $K_i$ the density $r(q|p)$ as all the points are reachable from any $p\in K_i$ with a trajectory with no reflections. The symmetry of the probability density function $r(q|p)$ follows from the uniformity of the distribution of the directions and reversibility of a billiard trajectory due to the reﬂection law: the angle of incidence equals the angle of reﬂection.

For the case of a geodesically non-convex $K_i$, the connectedness guarantees that starting from any point, we can reach a measurable neighborhood of any other point of $K_i$. Thus, there exists a piece-wise geodesic trajectory that connects any two points in $K_i$. Therefore, the transition probability density function $r(q|p) > 0$ for any $q,p\in K_i$. The symmetry of $r(q,p)$ holds using the same arguments as in the geodesically convex case. 
\end{proof}

%\section{Computing $\alpha_k$}\label{appnd:computing_ak}
\clearpage
\section{Tables}

% USA

\begin{table}[!h]
\centering
\footnotesize
\begin{tabular}{|c|cc|ccccc|}
 \hline
 \multicolumn{3}{|c|}{USA} & \makecell{overall\\market} & \makecell{small\\defensives} & \makecell{small\\cyclicals} & \makecell{large\\defensives} & \makecell{large\\cyclicals} \\
 \hline
  &  & Min Variance & 0.099 &  &  &  &  \\
 &  & CapW-BM & 0.085 &  &  &  &  \\
 \cline{2-8}
 \multirow{9}{*}{\begin{turn}{90}Annualized Returns\end{turn} } & \multirow{5}{*}{\begin{turn}{90}\makecell{Sampling-\\based}\end{turn} } & 1st volatility level & 0.097 & 0.087 & 0.109 & 0.067 & 0.101
 \\
 & & 2nd volatility level & 0.098 & 0.088 & 0.108 & 0.087 & 0.105  \\
 & & 3rd volatility level & 0.095 & 0.089 & 0.108 & 0.089 & 0.104 \\
 & & 4th volatility level & 0.095 & 0.079 & 0.105 & 0.085 & 0.096 \\
 & & 5th volatility level & 0.105 & 0.066 & 0.048	& 0.084 & 0.080 \\
 \cline{2-8}
  & \multirow{5}{*}{\begin{turn}{90}\makecell{Sorting-\\based}\end{turn} } & 1st volatility level & 0.094 & 0.079 & 0.113 & 0.074 & 0.099
 \\
 & & 2nd volatility level & 0.098 & 0.072 & 0.124 & 0.088 & 0.095 \\
 & & 3rd volatility level & 0.109 & 0.096 & 0.126 & 0.091 & 0.104 \\
 & & 4th volatility level & 0.072 & 0.080 & 0.031 & 0.064 & 0.084 \\
 & & 5th volatility level & 0.064 & 0.027 & 0.049 & 0.085 & 0.049 \\
 \hline
\end{tabular}
\caption{
\label{tab:annualized_returns_usa}}
\end{table}

\begin{table}[!h]
\centering
\footnotesize
\begin{tabular}{|c|cc|ccccc|}
 \hline
 \multicolumn{3}{|c|}{USA} & \makecell{overall\\market} & \makecell{small\\defensives} & \makecell{small\\cyclicals} & \makecell{large\\defensives} & \makecell{large\\cyclicals} \\
 \hline
  &  & Min Variance & 0.138 &  &  &  &  \\
 &  & CapW-BM & 0.191 &  &  &  &  \\
 \cline{2-8}
 \multirow{9}{*}{\begin{turn}{90} Annualized st.d.\end{turn} } & \multirow{5}{*}{\begin{turn}{90}\makecell{Sampling-\\based}\end{turn} } & 1st volatility level & 0.154 & 0.147 & 0.163 & 0.176 & 0.184 \\
 & & 2nd volatility level & 0.186 & 0.160 & 0.183 & 0.198 & 0.220 \\
 & & 3rd volatility level & 0.210 & 0.173 & 0.215 & 0.221 & 0.251 \\
 & & 4th volatility level & 0.230 & 0.202 & 0.269 & 0.240 & 0.291 \\
 & & 5th volatility leve & 0.268 & 0.259 & 0.363 & 0.294 & 0.375 \\
 \cline{2-8}
  & \multirow{5}{*}{\begin{turn}{90}\makecell{Sorting-\\based}\end{turn} } & 1st volatility level & 0.157 & 0.147 & 0.163 & 0.177 & 0.195 \\
 & & 2nd volatility level & 0.176 & 0.160 & 0.179 & 0.201 & 0.225 \\
 & & 3rd volatility level & 0.195 & 0.169 & 0.216 & 0.221 & 0.254 \\
 & & 4th volatility leve & 0.242 & 0.220 & 0.297 & 0.247 & 0.305 \\
 & & 5th volatility leve & 0.282 & 0.266 & 0.353 & 0.290 & 0.360 \\
 \hline
\end{tabular}
\caption{
\label{tab:annualized_standard_deviations_usa}}
\end{table}

\begin{table}[!h]
\centering
\footnotesize
\begin{tabular}{|c|cc|ccccc|}
 \hline
 \multicolumn{3}{|c|}{USA} & \makecell{overall\\market} & \makecell{small\\defensives} & \makecell{small\\cyclicals} & \makecell{large\\defensives} & \makecell{large\\cyclicals} \\
 \hline
  &  & Min Variance & 0.716 &  &  &  &  \\
 &  & CapW-BM & 0.445 &  &  &  &  \\
 \cline{2-8}
 \multirow{9}{*}{\begin{turn}{90} Annualized Sharpe Ratios\end{turn} } & \multirow{5}{*}{\begin{turn}{90}\makecell{Sampling-\\based}\end{turn} } & 1st volatility level & 0.630 & 0.595 & 0.673 & 0.380 & 0.548 \\
 & & 2nd volatility level & 0.525 & 0.549 & 0.590 & 0.441 & 0.479 \\
 & & 3rd volatility level & 0.455 & 0.514 & 0.504 & 0.405 & 0.415 \\
 & & 4th volatility level & 0.413 & 0.390 & 0.390 & 0.356 & 0.328 \\
 & & 5th volatility level & 0.391 & 0.256 & 0.133 & 0.287 & 0.214 \\
 \cline{2-8}
  & \multirow{5}{*}{\begin{turn}{90}\makecell{Sorting-\\based}\end{turn} } & 1st volatility level & 0.598 & 0.540 & 0.690 & 0.418 & 0.507 \\
 & & 2nd volatility level & 0.557 & 0.453 & 0.692 & 0.438 & 0.424 \\
 & & 3rd volatility level & 0.559 & 0.569 & 0.583 & 0.412 & 0.411 \\
 & & 4th volatility level & 0.299 & 0.364 & 0.106 & 0.260 & 0.275 \\
 & & 5th volatility level & 0.228 & 0.103 & 0.138 & 0.293 & 0.135 \\
 \hline
\end{tabular}
\caption{
\label{tab:annualized_sharpe_ratios_usa}}
\end{table}

% Europe

\begin{table}[!h]
\centering
\footnotesize
\begin{tabular}{|c|cc|ccccc|}
 \hline
 \multicolumn{3}{|c|}{Europe} & \makecell{overall\\market} & \makecell{small\\defensives} & \makecell{small\\cyclicals} & \makecell{large\\defensives} & \makecell{large\\cyclicals} \\
 \hline
  &  & Min Variance & 0.087 &  &  &  &  \\
 &  & CapW-BM & 0.059 &  &  &  &  \\
 \cline{2-8}
 \multirow{9}{*}{\begin{turn}{90} Annualized Returns \end{turn} } & \multirow{5}{*}{\begin{turn}{90}\makecell{Sampling-\\based}\end{turn} } & 1st volatility level & 0.081 & 0.092 & 0.115 & 0.062 & 0.060 \\
 & & 2nd volatility level & 0.079 & 0.079 & 0.097 & 0.073 & 0.078 \\
 & & 3rd volatility level & 0.074 & 0.073 & 0.087 & 0.063 & 0.072 \\
 & & 4th volatility level & 0.067 & 0.075 & 0.063 & 0.059 & 0.078 \\
 & & 5th volatility level & 0.065 & 0.070 & 0.033 & 0.053 & 0.060 \\
 \cline{2-8}
  & \multirow{5}{*}{\begin{turn}{90}\makecell{Sorting-\\based}\end{turn} } & 1st volatility level & 0.098 & 0.094 & 0.135 & 0.064 & 0.087 \\
  & & 2nd volatility level & 0.092 & 0.081 & 0.109 & 0.088 & 0.074\\
 & & 3rd volatility level & 0.075 & 0.073 & 0.089 & 0.059 & 0.060\\
 & & 4th volatility leve & 0.076 & 0.050 & 0.109 & 0.060 & 0.063\\
 & & 5th volatility leve & 0.018 & 0.046 & -0.042 & 0.005 & 0.025\\
 \hline
\end{tabular}
\caption{
\label{tab:annualized_returns_europe}}
\end{table}

\begin{table}[!h]
\centering
\footnotesize
\begin{tabular}{|c|cc|ccccc|}
 \hline
 \multicolumn{3}{|c|}{Europe} & \makecell{overall\\market} & \makecell{small\\defensives} & \makecell{small\\cyclicals} & \makecell{large\\defensives} & \makecell{large\\cyclicals} \\
 \hline
 &  & Min Variance & 0.159 &  &  &  &  \\
 &  & CapW-BM & 0.219 &  &  &  &  \\
 \cline{2-8}
 \multirow{9}{*}{\begin{turn}{90} Annualized st.d. \end{turn} } & \multirow{5}{*}{\begin{turn}{90}\makecell{Sampling-\\based}\end{turn} } & 1st volatility level &  0.170 & 0.183 & 0.177 & 0.200 & 0.189 \\
 & & 2nd volatility level & 0.197 & 0.190 & 0.192 & 0.224 & 0.222 \\
 & & 3rd volatility level & 0.222 & 0.199 & 0.212 & 0.250 & 0.238 \\
 & & 4th volatility level & 0.242 & 0.219 & 0.244 & 0.274 & 0.281 \\
 & & 5th volatility level & 0.278 & 0.248 & 0.347 & 0.326 & 0.344 \\
 \cline{2-8}
  & \multirow{5}{*}{\begin{turn}{90}\makecell{Sorting-\\based}\end{turn} } & 1st volatility level & 0.172 & 0.188 & 0.179 & 0.197 & 0.192 \\
 & & 2nd volatility level & 0.186 & 0.183 & 0.192 & 0.226 & 0.222 \\
 & & 3rd volatility level & 0.205 & 0.197 & 0.211 & 0.253 & 0.240 \\
 & & 4th volatility level & 0.235 & 0.219 & 0.252 & 0.284 & 0.270 \\
 & & 5th volatility level & 0.281 & 0.246 & 0.347 & 0.330 & 0.328 \\
 \hline
\end{tabular}
\caption{
\label{tab:annualized_standard_deviations_europe}}
\end{table}

\begin{table}[!h]
\centering
\footnotesize
\begin{tabular}{|c|cc|ccccc|}
 \hline
 \multicolumn{3}{|c|}{Europe} & \makecell{overall\\market} & \makecell{small\\defensives} & \makecell{small\\cyclicals} & \makecell{large\\defensives} & \makecell{large\\cyclicals} \\
 \hline
  &  & Min Variance & 0.543 &  &  &  &  \\
 &  & CapW-BM & 0.268 &  &  &  &  \\
 \cline{2-8}
 \multirow{9}{*}{\begin{turn}{90} Annualized Sharpe Ratios\end{turn} } & \multirow{5}{*}{\begin{turn}{90}\makecell{Sampling-\\based}\end{turn} } & 1st volatility level & 0.473 & 0.503 & 0.651 & 0.308 & 0.316 \\
 & & 2nd volatility level & 0.398 & 0.417 & 0.505 & 0.326 & 0.351 \\
 & & 3rd volatility level & 0.333 & 0.367 & 0.410 & 0.253 & 0.302 \\
 & & 4th volatility level & 0.279 & 0.344 & 0.258 & 0.217 & 0.275 \\
 & & 5th volatility level & 0.233 & 0.284 & 0.096 & 0.163 & 0.174 \\
 \cline{2-8}
  & \multirow{5}{*}{\begin{turn}{90}\makecell{Sorting-\\based}\end{turn} } & 1st volatility level & 0.573 & 0.503 & 0.755 & 0.327 & 0.456 \\
 & & 2nd volatility level & 0.496 & 0.442 & 0.568 & 0.389 & 0.335 \\
 & & 3rd volatility level & 0.368 & 0.371 & 0.422 & 0.234 & 0.249 \\
 & & 4th volatility level & 0.325 & 0.229 & 0.430 & 0.213 & 0.234 \\
 & & 5th volatility level & 0.065 & 0.187 & -0.122 & 0.016 & 0.075 \\
 \hline
\end{tabular}
\caption{
\label{tab:annualized_sharpe_ratios_europe}}
\end{table}

\begin{table}[t]
\centering
\footnotesize
\begin{tabular}{|c|ccccc|}\hline
 \multicolumn{6}{|c|}{USA} \\
 \hline
 Volatility level &  1st   &   2nd    &  3rd    &  4th  & 5th \\
\hline
\makecell{Model-based\\ mode} &  0.051   &    0.042   &  0.071   &  0.057  & 0.061 \\
\hdashline
\makecell{Sampling-based\\average} &  0.051   &    0.041   &   0.071   &  0.056  & 0.060 \\
\hdashline
\makecell{Sorting-based\\(quintile portfolios)} &  0.002   &   0.010    & 0.004  &  0  & 0.001 \\
\hline\hline
 \multicolumn{6}{|c|}{Europe} \\
 \hline
 Volatility level &  1st   &   2nd    &  3rd    &  4th  & 5th \\
\hline
\makecell{Model-based\\ mode} &  0.049   &    0.045   &   0.054   &   0.061 & 0.074 \\
\hdashline
\makecell{Sampling-based\\average} &  0.048   &    0.044   &     0.053 &  0.056  &  0.073 \\
\hdashline
\makecell{Sorting-based\\(quintile portfolios)} & 0    &    0   &   0   &   0.051 & 0.007 \\
\hline
\end{tabular}
\caption{For each volatility cluster in Figure~\ref{fig:risk_return_usa_europe} we fit a log-concave distribution using the non-parametric model in~\cite{Cule10}. We report the probability ---w.r.t.\ the log-concave measure we obtain from {\tt LogConcDEAD} package~\cite{Culejss09}--- of a small rectangle centered on the model-based mode, the sampling-based average and the sorting-based quintile portfolio. In all cases, the volume of the rectangle is $1\%$ of the volume of the support of the distribution obtained by the log-concave model. See Fig.~\ref{fig:dist_usa},~\ref{fig:dist_europe} (Appendix) for a visualization. \label{tab:quintile_mode_prob}}
\end{table}

%\clearpage

\clearpage
\section{Figures}\label{appndx:figures}

\begin{figure}[!h] 
\centering
\includegraphics[width=0.32\linewidth]{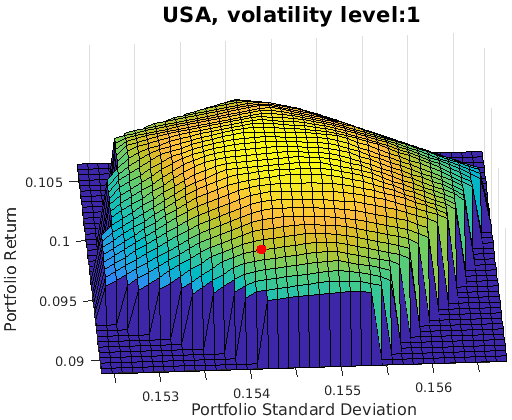}
\includegraphics[width=0.32\linewidth]{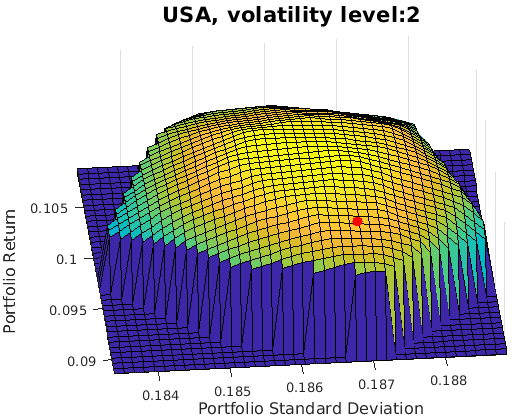}
\includegraphics[width=0.32\linewidth]{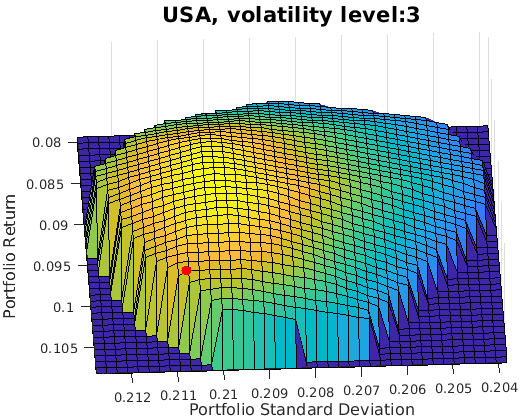}\\
\vspace{0.4cm}
\includegraphics[width=0.35\linewidth]{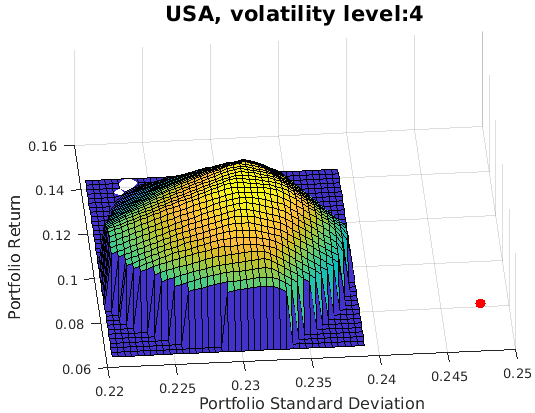}
\includegraphics[width=0.35\linewidth]{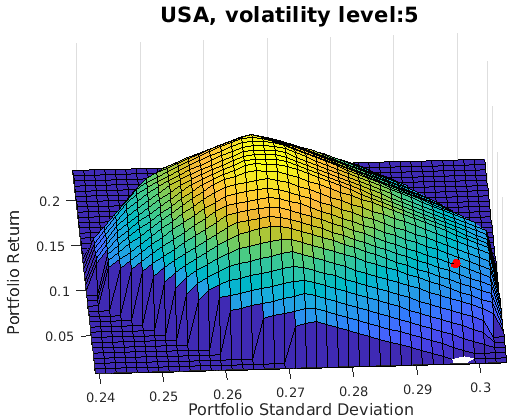}
\caption{For each volatility cluster in the left plot in Figure~\ref{fig:risk_return_usa_europe} ---U.S. companies--- we fit a log-concave distribution using the non-parametric model in~\cite{Cule10}. With a red dot we report the sorting-based quintile portfolio.}
\label{fig:dist_usa}
\end{figure}

\begin{figure}[!h] 
\centering
\includegraphics[width=0.32\linewidth]{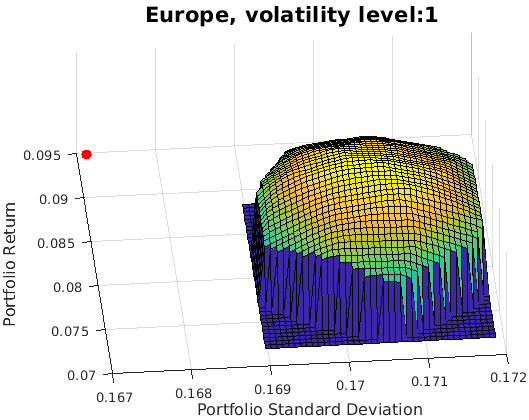}
\includegraphics[width=0.32\linewidth]{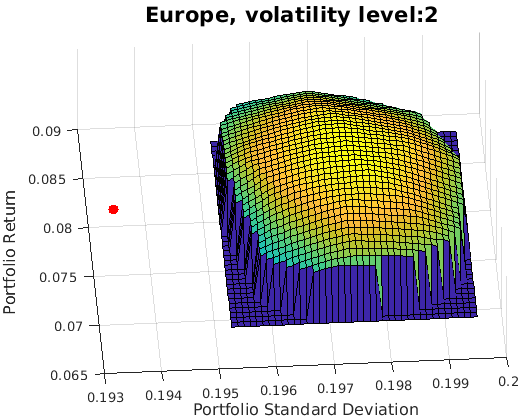}
\includegraphics[width=0.32\linewidth]{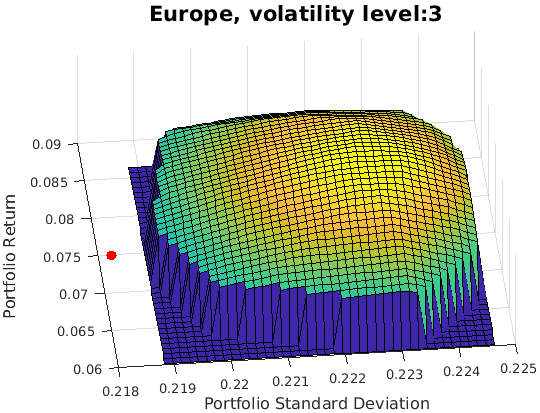}\\
\vspace{0.4cm}
\includegraphics[width=0.35\linewidth]{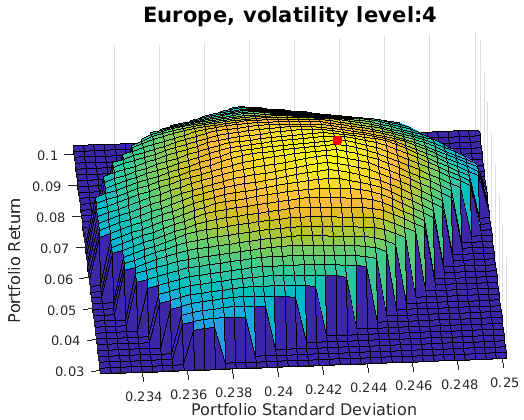}
\includegraphics[width=0.35\linewidth]{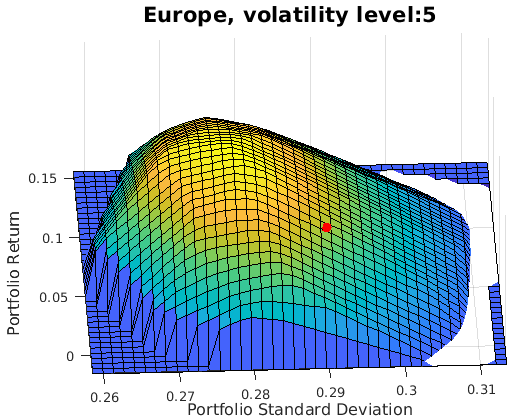}
\caption{For each volatility cluster in the right plot in Figure~\ref{fig:risk_return_usa_europe} ---European companies--- we fit a log-concave distribution using the non-parametric model in~\cite{Cule10}. With a red dot we report the sorting-based quintile portfolio.}
\label{fig:dist_europe}
\end{figure}

\if 0
\clearpage
\section{Copulas}

\begin{figure}[!h] 
\centering
\includegraphics[width=0.32\linewidth]{./figures_copulas/cop_usa_vola_1.png}
\includegraphics[width=0.32\linewidth]{./figures_copulas/cop_usa_vola_2.png}
\includegraphics[width=0.32\linewidth]{./figures_copulas/cop_usa_vola_3.png}\\
\vspace{0.4cm}
\includegraphics[width=0.32\linewidth]{./figures_copulas/cop_usa_vola_4.png}
\includegraphics[width=0.32\linewidth]{./figures_copulas/cop_usa_vola_5.png}
\includegraphics[width=0.32\linewidth]{./figures_copulas/cop_usa_concat.png}
\caption{For each volatility cluster in the left plot in Figure~\ref{fig:risk_return_europe_ex_ch_and_usa} ---U.S.\ companies--- we estimate the copula between risk and return.}
\label{fig:cop_usa}
\end{figure}

\begin{figure}[!h] 
\centering
\includegraphics[width=0.32\linewidth]{./figures_copulas/cop_europe_vola_1.png}
\includegraphics[width=0.32\linewidth]{./figures_copulas/cop_europe_vola_2.png}
\includegraphics[width=0.32\linewidth]{./figures_copulas/cop_europe_vola_3.png}\\
\vspace{0.4cm}
\includegraphics[width=0.32\linewidth]{./figures_copulas/cop_europe_vola_4.png}
\includegraphics[width=0.32\linewidth]{./figures_copulas/cop_europe_vola_5.png}
\includegraphics[width=0.32\linewidth]{./figures_copulas/cop_europe_concat.png}
\caption{For each volatility cluster in the right plot in Figure~\ref{fig:risk_return_europe_ex_ch_and_usa} ---European companies--- we estimate the copula between risk and return.}
\label{fig:cop_europe}
\end{figure}
\fi

\end{document}